\documentclass{article}

\usepackage{amssymb,amsmath,xspace,enumerate,gensymb}
\usepackage{amsthm}

\usepackage[pagewise]{lineno}
\usepackage{xspace}
\usepackage{arydshln}
\usepackage{rotating}
\usepackage{listings}
\usepackage{soul}
\usepackage{ifmtarg}
\usepackage{stmaryrd} \usepackage{paralist}
\usepackage{chngcntr}

\usepackage{tikz}
\usetikzlibrary{arrows,decorations.pathmorphing,backgrounds,calc,positioning,fit,petri,matrix,backgrounds}
\usepackage{pgffor}
\usepackage{extarrows}
\usepackage{appendix}

\usepackage[dvips,colorlinks=true]{hyperref}
\usepackage{bold-extra}
\newif\ifcomments
\newif\ifchanges
\commentsfalse\changesfalse
\commentstrue\changestrue 
\usepackage{authblk}

\begin{document}

\newcommand  {\myclass} [1]  {\ensuremath{\textsc{#1}}}

\newcommand{\StaClass}[1]{\myclass{#1}\xspace}

\newcommand{\DynClass}[1]{\myclass{Dyn#1}\xspace}
\newcommand{\dDynClass}[1]{\myclass{$\Delta$-Dyn#1}\xspace}

\newcommand  {\myproblem} [1] {\textsc{#1}}

\newcommand  {\problemdescr} [3] {
    \vspace{3mm}
    \def\Name{#1}
    \def\Input{#2}
    \def\Question{#3}
     \hspace{5mm}\begin{tabular}{r p{\columnWidth}r}      \textit{Query:} & \myproblem{\Name} \\
      \textit{Input:} & \Input \\
      \textit{Question:} & \Question
     \end{tabular}
    \vspace{3mm}
    }

\newcommand  {\dynproblemdescr} [4] {
    \vspace{3mm}
    \def\Name{#1}
    \def\Input{#2}
    \def\Updates{#3}  
    \def\Question{#4}
    \hspace{5mm}\begin{tabular}{r p{\columnWidth}r}      \textit{Query:} & \myproblem{\Name} \\
      \textit{Input:} & \Input \\
      \textit{Question:} & \Question
    \end{tabular}
    \vspace{3mm}
}
\newcommand{\dynProbDescr}[4]{\dynproblemdescr{#1}{#2}{#3}{#4}}

\newcommand  {\problem}[1] {\myproblem{#1}}

\newcommand{\dynProb}[1] {\myproblem{Dyn(#1)}}

\newcommand{\class}{\calC}

\newcommand  {\TIME}    {\myclass{TIME}}
\newcommand  {\DTIME}   {\myclass{DTIME}}
\newcommand  {\NTIME}   {\myclass{NTIME}}
\newcommand  {\ATIME}   {\myclass{ATIME}}
\newcommand  {\SPACE}   {\myclass{SPACE}}
\newcommand  {\DSPACE}   {\myclass{DSPACE}}
\newcommand  {\NSPACE}  {\myclass{NSPACE}}
\newcommand  {\coNSPACE}        {\myclass{coNSPACE}}

\newcommand     {\LOGCFL}     {\myclass{LOGCFL}}
\newcommand     {\LOGDCFL}     {\myclass{LOGDCFL}}
\newcommand     {\LOGSPACE}     {\myclass{LOGSPACE}}
\newcommand     {\NLOGSPACE}     {\myclass{NLOGSPACE}}
\newcommand     {\classL}   {\myclass{L}}
\newcommand     {\NL}   {\myclass{NL}}
\newcommand     {\coNL}   {\myclass{coNL}}
\renewcommand   {\P}    {\myclass{P}}
\newcommand     {\myP}    {\myclass{P}}
\newcommand     {\PTIME}    {\myclass{PTIME}}
\newcommand     {\NP}   {\myclass{NP}}
\newcommand     {\NPC}   {\myclass{NPC}}
\newcommand     {\PH}   {\myclass{PH}}
\newcommand     {\coNP} {\myclass{coNP}}
\newcommand     {\NPSPACE}      {\myclass{NPSPACE}}
\newcommand     {\PSPACE}       {\myclass{PSPACE}}
\newcommand     {\IP}   {\myclass{IP}}
\newcommand     {\POLYLOGSPACE} {\myclass{POLYLOGSPACE}}
\newcommand     {\DET}  {\myclass{DET}}
\newcommand     {\EXP}  {\myclass{EXP}}
\newcommand     {\NEXP}  {\myclass{NEXP}}
\newcommand     {\EXPTIME}  {\myclass{EXPTIME}}
\newcommand     {\TWOEXPTIME}  {\myclass{2-EXPTIME}}
\newcommand     {\TWOEXP}  {\myclass{2-EXP}}
\newcommand     {\NEXPTIME}  {\myclass{NEXPTIME}}
\newcommand     {\coNEXPTIME}  {\myclass{coNEXPTIME}}
\newcommand     {\EXPSPACE}  {\myclass{EXPSPACE}}
\newcommand     {\RP}   {\myclass{RP}}
\newcommand     {\RL}   {\myclass{RL}}
\newcommand     {\coRP} {\myclass{coRP}}
\newcommand     {\ZPP}  {\myclass{ZPP}}
\newcommand     {\BPP}  {\myclass{BPP}}
\newcommand     {\PP}   {\myclass{PP}}
\newcommand     {\NC}   {\myclass{NC}}
\newcommand     {\SAC}   {\myclass{SAC}}
\newcommand     {\ACC}   {\myclass{ACC}}
\newcommand     {\tc}   {\myclass{TC}}   \newcommand     {\PPoly}{\myclass{\mbox{P}/\mbox{Poly}}} 
\newcommand     {\FOarb}   {\myclass{FO(arb)}}

\newcommand     {\NLIN}   {\myclass{NLIN}}
\newcommand     {\DLIN}   {\myclass{DLIN}}

\newcommand  {\APTIME}   {\myclass{APTIME}}
\newcommand  {\ALOGSPACE}   {\myclass{ALOGSPACE}}

\newcommand{\FO}{\StaClass{FO}}
\newcommand{\MSO}[1][\quant]{\StaClass{MSO}}
\newcommand{\EMSO}{\StaClass{$\exists$MSO}}
\newcommand{\QFO}[1][\quant]{\StaClass{\ensuremath{#1}FO}}
\newcommand{\cQFO}[1][\quant]{\StaClass{\ensuremath{\overline{#1}}FO}}
\newcommand{\EFO}{\QFO[\exists^*]}
\newcommand{\AFO}{\QFO[\forall^*]}
\newcommand{\AEFO}{\StaClass{$\forall/\exists$FO}}
\newcommand{\CQ}[1][]{\StaClass{CQ}}
\newcommand{\UCQ}[1][]{\StaClass{UCQ}}
\newcommand{\CQneg}[1][]{\StaClass{CQ\ensuremath{^{\mneg}}}}
\newcommand{\UCQneg}[1][]{\StaClass{UCQ\ensuremath{^{\mneg}}}}
\newcommand{\Prop}{\StaClass{Prop}}
\newcommand{\QF}{\StaClass{QF}}
\newcommand{\PropCQ}{\StaClass{PropCQ}}
\newcommand{\PropUCQ}{\StaClass{PropUCQ}}
\newcommand{\PropCQneg}{\StaClass{PropCQ{\ensuremath{^{\mneg}}}}}
\newcommand{\PropUCQneg}{\StaClass{PropUCQ{\ensuremath{^{\mneg}}}}}

\newcommand{\mneg}{\neg} 
\newcommand{\DynTC}{\DynClass{TC}}

\newcommand{\DynProp}{\DynClass{Prop}}
\newcommand{\DynProj}{\DynClass{Projections}}
\newcommand{\DynQF}{\DynClass{QF}}
\newcommand{\DynFO}{\DynClass{FO}}
\newcommand{\DynFOpos}{\DynClass{FO$^{\wedge, \vee}$}}
\newcommand{\DynFOand}{\DynClass{FO$^{\wedge}$}}

\newcommand{\DynC}{\DynClass{$\class$}}
\newcommand{\DynUCQ}{\DynClass{UCQ}}
\newcommand{\DynCQ}{\DynClass{CQ}}
\newcommand{\DynUCQneg}{\DynClass{UCQ$^\mneg$}}
\newcommand{\DynCQneg}{\DynClass{CQ$^\mneg$}}
\newcommand{\DynCQPM}{\DynCQneg}
\newcommand{\DyncQFO}{\DynClass{$\cquant$FO}}

\newcommand{\DynQFO}[1][\quant]{\DynClass{\QFO[#1]}}
\newcommand{\DynEFO}{\DynQFO[\exists^*]}
\newcommand{\DynAFO}{\DynQFO[\forall^*]}

\newcommand{\DynAEFO}{\DynClass{$\forall/\exists$FO}}
\newcommand{\DynAND}{\DynClass{PropCQ}}
\newcommand{\DynAnd}{\DynAND}
\newcommand{\DynPropCQ}{\DynAND}

\newcommand{\DynPropPos}{\DynClass{PropUCQ}}
\newcommand{\DynPropAO}{\DynPropPos}
\newcommand{\DynPropUCQ}{\DynPropPos}

\newcommand{\DynAndNeg}{\DynClass{PropCQ{\ensuremath{^{\mneg}}}}}
\newcommand{\DynPropCQneg}{\DynAndNeg}
\newcommand{\DynPropUCQneg}{\DynClass{PropUCQ{\ensuremath{^{\mneg}}}}}

\newcommand{\DynOrNeg}{\DynClass{Or{\ensuremath{^{\mneg}}}}}

\newcommand{\dDynProp}{\dDynClass{Prop}}
\newcommand{\dDynPropPos}{\dDynClass{PropUCQ}}
\newcommand{\dDynAndOr}{\dDynPropPos}
\newcommand{\dDynQF}{\dDynClass{QF}}
\newcommand{\dDynFO}{\dDynClass{FO}}
\newcommand{\dDynFOpos}{\dDynClass{FO$^{\wedge, \vee}$}}
\newcommand{\dDynFOand}{\dDynClass{FO$^{\wedge}$}}
\newcommand{\dDynC}{\dDynClass{$\class$}}
\newcommand{\dDynUCQ}{\dDynClass{UCQ}}
\newcommand{\dDynCQ}{\dDynClass{CQ}}
\newcommand{\dDynUCQneg}{\dDynClass{UCQ$^\mneg$}}
\newcommand{\dDynCQneg}{\dDynClass{CQ$^\mneg$}}
\newcommand{\dDynCQPM}{\dDynCQneg}

\newcommand{\dDynQFO}[1][\quant]{\dDynClass{\QFO[#1]}}
\newcommand{\dDynEFO}{\dDynQFO[\exists^*]}
\newcommand{\dDynAFO}{\dDynQFO[\forall^*]}

\newcommand{\dDynAEFO}{\dDynClass{$\forall/\exists$FO}}
\newcommand{\dDynAND}{\dDynPropCQ}
\newcommand{\dDynAnd}{\dDynAND}
\newcommand{\dDynConj}{\dDynClass{Conj}}
\newcommand{\dDynPropAO}{\dDynClass{Prop$^{\wedge, \vee}$}}
\newcommand{\dDyncQFO}{\dDynClass{$\cquant$FO}}
\newcommand{\dDynPropUCQneg}{\dDynClass{PropUCQ{\ensuremath{^{\mneg}}}}}
\newcommand{\dDynPropUCQ}{\dDynClass{PropUCQ}}
\newcommand{\dDynPropCQneg}{\dDynClass{PropCQ{\ensuremath{^{\mneg}}}}}
\newcommand{\dDynPropCQ}{\dDynClass{PropCQ}}

\newcommand{\equalcardinality}{\textsc{EqualCardinality}\xspace}\newcommand{\reach}{\textsc{Reach}\xspace}\newcommand{\altreach}{\textsc{Alt-Reach}\xspace}
\newcommand{\stgraph}{$s$-$t$-graph\xspace}
\newcommand{\stgraphs}{$s$-$t$-graphs\xspace}
\newcommand{\reachQ}{\textsc{Reach}\xspace}
\newcommand{\streachQ}{\textsc{$s$-$t$-Reach}\xspace}
\newcommand{\streachabilityquery}{$s$-$t$-reachability query\xspace}
\newcommand{\stTwoPath}{\problem{$s$-$t$-Two\-Path}\xspace}
\newcommand{\sTwoPath}{\problem{$s$-Two\-Path}\xspace}
\newcommand{\clique}[1]{\problem{$#1$-Clique}\xspace}
\newcommand{\colorability}[1]{\problem{$#1$-Col}\xspace}
\newcommand{\streach}{$s$-$t$-Reach}
\newcommand{\streachp}{\problem{\streach}\xspace}
\newcommand{\layeredstreach}[1]{#1-Layered-$s$-$t$-Reach}
\newcommand{\layeredstreachp}[1]{\problem{\layeredstreach{#1}}\xspace}

\newcommand{\dynClique}[1]{\dynProb{$#1$-Clique}\xspace}
\newcommand{\dynColorability}[1]{\dynProb{$#1$-Col}\xspace}

\newcommand{\probEqualCardinalityText}{EqualCardinality}
\newcommand{\EqualCardinality}{\problem{\probEqualCardinalityText}\xspace}
\newcommand{\EqualCardinalityDescr}{\problemdescr{\probEqualCardinalityText}{Unary relations $A$ and $B$}{Do $A$ and $B$ have the same cardinality?\xspace}}

\newcommand{\dynEqualCardinality}{\dynProb{\probEqualCardinalityText}\xspace}
\newcommand{\dynEqualCardinalityDescr}{\dynProbDescr{\probEqualCardinalityText}{Unary relations $A$ and $B$}{Element insertions and deletions}{Do $A$ and $B$ have the same cardinality?\xspace}}

\newcommand{\dynReachQ}{\dynProb{\textsc{Reach}}\xspace}
\newcommand{\dynstReachQ}{\dynProb{\textsc{$s$-$t$-Reach}}\xspace}

\newcommand{\dynstTwoPath}{\dynProb{\stTwoPath}\xspace}
\newcommand{\dynsTwoPath}{\dynProb{\sTwoPath}\xspace}

\newcommand{\dynlayeredstreach}[1]{Dyn-#1-Layered-$s$-$t$-Reach}
\newcommand{\dynlayeredstreachp}[1]{\problem{\dynlayeredstreach{#1}}\xspace}

\newcommand{\mtext}[1]{\textsc{#1}}

\providecommand {\calA}      {{\mathcal A}\xspace}
\providecommand {\calB}      {{\mathcal B}\xspace}
\providecommand {\calC}      {{\mathcal C}\xspace}
\providecommand {\calD}      {{\mathcal D}\xspace}
\providecommand {\calE}      {{\mathcal E}\xspace}
\providecommand {\calF}      {{\mathcal F}\xspace}
\providecommand {\calG}      {{\mathcal G}\xspace}
\providecommand {\calH}      {{\mathcal H}\xspace}
\providecommand {\calK}      {{\mathcal K}\xspace}
\providecommand {\calI}      {{\mathcal I}\xspace}
\providecommand {\calL}      {{\mathcal L}\xspace}
\providecommand {\calM}      {{\mathcal M}\xspace}
\providecommand {\calN}      {{\mathcal N}\xspace}
\providecommand {\calO}      {{\mathcal O}\xspace}
\providecommand {\calP}      {{\mathcal P}\xspace}
\providecommand {\calQ}      {{\mathcal Q}\xspace}
\providecommand {\calR}      {{\mathcal R}\xspace}
\providecommand {\calS}      {{\mathcal S}\xspace}
\providecommand {\calT}      {{\mathcal T}\xspace}
\providecommand {\calU}      {{\mathcal U}\xspace}
\providecommand {\calV}      {{\mathcal V}\xspace}
\providecommand {\calX}      {{\mathcal X}\xspace}
\providecommand {\calZ}      {{\mathcal Z}\xspace}

\renewcommand{\hat}[1]{\widehat{#1}}

\newcommand{\Ah}{\hat{A}}
\newcommand{\Bh}{\hat{B}}
\newcommand{\Ch}{\hat{C}}
\newcommand{\Dh}{\hat{D}}
\newcommand{\Eh}{\hat{E}}
\newcommand{\Fh}{\hat{F}}
\newcommand{\Gh}{\hat{G}}
\newcommand{\Hh}{\hat{H}}
\newcommand{\Ih}{\hat{I}}
\newcommand{\Jh}{\hat{J}}
\newcommand{\Kh}{\hat{K}}
\newcommand{\Lh}{\hat{L}}
\newcommand{\Mh}{\hat{M}}
\newcommand{\Nh}{\hat{N}}
\newcommand{\Oh}{\hat{O}}
\newcommand{\Ph}{\hat{P}}
\newcommand{\Qh}{\hat{Q}}
\newcommand{\Rh}{\hat{R}}
\newcommand{\Sh}{\hat{S}}
\newcommand{\Th}{\hat{T}}
\newcommand{\Uh}{\hat{U}}
\newcommand{\Vh}{\hat{V}}
\newcommand{\Wh}{\hat{W}}
\newcommand{\Xh}{\hat{X}}
\newcommand{\Yh}{\hat{Y}}
\newcommand{\Zh}{\hat{Z}}

\newcommand{\Psih}{\hat{\Psi}}
\newcommand{\psih}{\hat{\psi}}
\newcommand{\Phih}{\hat{\Phi}}
\newcommand{\phih}{\hat{\phi}}
\newcommand{\varphih}{\hat{\varphi}}

\newcommand{\N}{\ensuremath{\mathbb{N}}}

\newcommand{\Q}{\ensuremath{\mathbb{Q}}}

\newcommand{\R}{\ensuremath{\mathbb{R}}}

\newcommand{\perm}{\ensuremath{\pi}}

\newcommand{\bigO}{\ensuremath{O}}
\newcommand{\smallO}{\ensuremath{o}}

\newcommand{\allsubsets}[2]{[#1]^{#2}}

\newcommand{\pvec}[1]{\vec{#1}\mkern2mu\vphantom{#1}}

\newcommand{\kexp}[2]{\ensuremath{\exp^{#1}\hspace{-0.5mm}(#2)}}

\newcommand{\tower}[2]{\ensuremath{\text{tow}_{#1}\hspace{-0.5mm}(#2)}}

\newcommand{\klog}[2]{\ensuremath{\log^{#1}{\hspace{-0.5mm}(#2)}}}

\newcommand{\subs}{\sqsubseteq}

\newcommand{\disjointunion}{\uplus}

\providecommand{\power}[1]{\ensuremath{\calP(#1)}\xspace}

\newcommand{\restrict}[2]{#1\mspace{-3mu}\upharpoonright \mspace{-3mu}#2}

\newcommand{\isomorph}{\simeq}
\newcommand{\isomorphVia}[1]{\isomorph_{#1}}
\newcommand{\swap}[2]{id{[#1, #2]}}

\newcommand{\df}{\ensuremath{\mathrel{\smash{\stackrel{\scriptscriptstyle{
    \text{def}}}{=}}}} \;}

\newcommand{\refeq}[1]{\ensuremath{{\stackrel{\scriptstyle{
    \text{#1}}}{=}}}}

\newcommand{\longlongeq}{=\joinrel=\joinrel=\joinrel=}
\newcommand{\longeq}{=\joinrel=\joinrel=}
\newcommand{\reflongeq}[1]{\ensuremath{{\stackrel{\scriptstyle{
    \text{#1}}}{\longeq}}}}

\newcommand{\ramseyw}[1]{\ensuremath{R_{#1}}}

\makeatletter \newcommand{\auxramsey}[4]{
  \@ifmtarg{#1}{
    \@ifmtarg{#4}{
      \ensuremath{R(#2; #3)}
    }{
      \ensuremath{R^#4(#2; #3)}
    }
   }{
    \@ifmtarg{#4}{
      \ensuremath{R_{#1}(#2; #3)}
    }{
      \ensuremath{R^#4_{#1}(#2; #3)}
    }
  }
}

\newcommand{\ramsey}[3]{\auxramsey{#1}{#2}{#3}{}}
\newcommand{\homramsey}[2]{\auxramsey{}{#2}{#1}{\text{hom}}}
\newcommand{\mfoldramsey}[3]{\auxramsey{}{#2}{#1}{#3}}

\newcommand{\norder}{\prec}

\newcommand{\col}{col}

\newcommand{\property}{($\ast$)}

\newcommand{\subseq}{\sqsubseteq}

\newcommand{\derive}{\Rightarrow}
\newcommand{\rmapsto}{\rightarrow}

\newcommand{\lpath}[1][]{\ensuremath{\mathrel{\smash{\stackrel{\scriptscriptstyle{
    #1}}{\rightsquigarrow}}}}}

  \newtheorem{theorem}{Theorem}[section]
   \newtheorem{lemma}[theorem]{Lemma}
   \newtheorem{corollary}[theorem]{Corollary}
   \newtheorem{proposition}[theorem]{Proposition}
   \newtheorem{claim}{Claim}

    \newtheorem{goal}{Goal}
    \theoremstyle{definition}
    \newtheorem{definition}{Definition}
    \newtheorem{example}{Example}
    
    \newtheorem {openquestion}{Open question}
    \newtheorem {question}{Question}
    \newtheorem {mainquestion}{Main question}

    \newenvironment{proofsketch}{\noindent\textsc{Proof sketch.}\enspace}{\qed \\}
    \newenvironment{proofidea}{\noindent\textsc{Proof idea.}\enspace}{\qed \\}

    \newenvironment{proofof}[1]{\vspace{2mm}\noindent\emph{Proof (of #1).}\enspace}{\qed\vspace{2mm}}
   \newenvironment{proofsketchof}[1]{\vspace{2mm}\noindent\emph{Proof sketch (of #1).}\enspace}{\qed\vspace{2mm}}

\newcommand{\eval}[3]{#1(#2/#3)}

\newcommand{\assignment}{\theta}

\newcommand{\arity}{\ensuremath{\text{Ar}}}
\newcommand{\arityFun}{\ensuremath{Ar_{\text{fun}}}}

\newcommand{\schema}{\tau}
\newcommand{\schemah}{\hat{\schema}}
\newcommand{\relSchema}{\schema_{\text{rel}}}
\newcommand{\relSchemah}{\schemah_{\text{rel}}}
\newcommand{\conSchema}{\schema_{\text{const}}}
\newcommand{\conSchemah}{\schemah_{\text{const}}}
\newcommand{\funSchema}{\schema_{\text{fun}}}
\newcommand{\funSchemah}{\schemah_{\text{fun}}}
\newcommand{\Terms}[2]{\textsc{Terms}^{#2}_{#1}}

\newcommand{\unaryTypes}[1]{\mathcal{UN}_{#1}}
\newcommand{\binaryTypes}[1]{\mathcal{BIN}_{#1}}
\newcommand{\naryTypes}[2]{\mathfrak{T}_{#1,#2}}

\newcommand{\nb}[3]{\calN_{#2}^{#3}(#1)}
\newcommand{\nbv}[3]{\vec \calN_{#2}^{#3}(#1)}

\newcommand{\mthen}{\rightarrow}
\newcommand{\mand}{\wedge}
\newcommand{\mor}{\vee}
\newcommand{\munion}{\cup}
\newcommand{\mintersect}{\cap}
\newcommand{\mdisjunion}{\biguplus}
\newcommand{\sem}[2]{\ensuremath{\llbracket #1\rrbracket_{#2}}} 
\newcommand{\arb}{\ensuremath{\star}}\newcommand{\generic}{\textsc{generic}}
\newcommand{\quant}{\mathbb{Q}}
\newcommand{\cquant}{\overline{\mathbb{Q}}}

\newcommand{\nd}{d}
\newcommand{\formulas}{\calC}
\newcommand{\symneg}[1]{\widehat{#1}}

\newcommand{\type}[2]{\ensuremath{\langle #1, #2 \rangle}}
\newcommand{\stype}[3]{\ensuremath{\langle #1, #2 \rangle_{#3}}}

\newcommand{\behaveEqual}[1]{\approx_{#1}}

\newcommand{\types}[2]{types_{#1}(#2)}
\newcommand{\numTypes}[2]{|\types{{#1}}{#2}|}

\newcommand{\eqtype}{\epsilon}

\newcommand{\struc}{\calS}
\newcommand{\db}{\calD}
\newcommand{\inp}{\calI}
\newcommand{\aux}{\calA}
\newcommand{\builtin}{\calB}
\newcommand{\domain}{D}

\newcommand{\query}{\calQ}
\newcommand{\cq}{\calC}

\newcommand{\querys}{Q}

\newcommand{\ans}[2]{\mtext{ans}(#1, #2)}

\newcommand{\updates}{\ensuremath{\Delta}}
\newcommand{\abstrDel}{\ensuremath{\updates_{Del}}}
\newcommand{\abstrIns}{\ensuremath{\updates_{Ins}}}
\newcommand{\abstrUpd}{\ensuremath{\updates}}

\newcommand{\init}{\mtext{Init}\xspace}

\newcommand{\ins}{\mtext{ins}}
\newcommand{\del}{\mtext{del}}

\newcommand{\insertdescr}[2]{\textbf{Insertion of \ensuremath{#2} into \ensuremath{#1}.}}
\newcommand{\deletedescr}[2]{\textbf{Deletion of \ensuremath{#2} from \ensuremath{#1}.}}

\newcommand{\state}{\struc}

\newcommand{\inpSchema}{\schema_{\text{in}}}
\newcommand{\auxSchema}{\schema_{\text{aux}}}
\newcommand{\eqSchema}{\schema_{=}}
\newcommand{\builtinSchema}{\schema_{\text{bi}}}

\newcommand{\auxInit}{\init_{\text{aux}}}
\newcommand{\builtinInit}{\init_{\text{bi}}}

\providecommand{\prog}{\ensuremath{\calP}\xspace}
\newcommand{\progb}{\ensuremath{Q}\xspace}

\newcommand{\updateDB}[2]{\ensuremath{#1(#2)}}
\newcommand{\updateState}[3]{\ensuremath{#1_{#2}(#3)}}
\newcommand{\updateRelation}[4]{\restrict{\ensuremath{{#1}_{#2}(#3)}}{#4}}

\newcommand{\transition}[3]{\ensuremath{{#1} \xrightarrow{#2}{#3}}}

\makeatletter \newcommand{\uf}[4]{
  \@ifmtarg{#4}{
    \ensuremath{\phi^{#1}_{#2}(#3)}
   }{
    \ensuremath{\phi^{#1}_{#2}(#3; #4)}
  }
}
\newcommand{\huf}[4]{
  \@ifmtarg{#4}{
    \ensuremath{\widehat{\phi}^{#1}_{#2}(#3)}
   }{
    \ensuremath{\widehat{\phi}^{#1}_{#2}(#3; #4)}
  }
}

\newcommand{\ufb}[4]{
  \@ifmtarg{#4}{
    \ensuremath{\psi^{#1}_{#2}(#3)}
   }{
    \ensuremath{\psi^{#1}_{#2}(#3; #4)}
  }
}

\newcommand{\ufbwa}[2]{
  \ensuremath{\psi^{#1}_{#2}}
}

\newcommand{\ufwa}[2]{
  \ensuremath{\phi^{#1}_{#2}}
}

    \makeatletter   \newcommand{\ufsubstitute}[5]{
    \@ifmtarg{#5}{
      \ensuremath{\phi^{#2}_{#3}[#1](#4)}
    }{
      \ensuremath{\phi^{#2}_{#3}[#1](#4; #5)}
    }
  }

    \makeatletter   \newcommand{\ufsubstitutewa}[3]{
      \ensuremath{\phi^{#2}_{#3}[#1]}
  }
  \makeatletter   \newcommand{\substitutewa}[2]{
      \ensuremath{#1[#2]}
  }

\newcommand{\ut}[4]{
  \ensuremath{t^{#1}_{#2}(#3; #4)}
}

\newcommand{\utw}[3]{
  \ensuremath{t^{#1}_{#2}(#3)}
}

\newcommand{\utwa}[2]{\ensuremath{t^{#1}_{#2}}}
\newcommand{\ite}[3]{
  \@ifmtarg{#1}{
    \ensuremath{\mtext{ite}}
   }{
    \mtext{ite}(#1,#2,#3)  
  }
}

\newcommand{\itewa}{
    \ensuremath{\mtext{ite}}
}

\providecommand{\nc}{\newcommand}
\providecommand{\rnc}{\renewcommand}
\providecommand{\pc}{\providecommand}

\renewcommand{\labelenumi}{(\alph{enumi})}

\newcommand{\Erdos}{Erd\H{o}s}

\ifcomments
\nc{\commentbox}[1]{\noindent\framebox{\parbox{\linewidth}{#1}}}
\nc{\todo}[1]{\ \\ {\color{red} \fbox{\parbox{\linewidth}{{\sc
          ToDo}:\\  #1}}}}

\setlength{\marginparwidth}{2.5cm}
\setlength{\marginparsep}{3pt}

\newcounter{CommentCounter}
\newcommand{\acomment}[2]{\ \\ \fbox{\parbox{\linewidth}{{\sc #1}: #2}}}
\newcommand{\mcomment}[2]{{\color{blue}(#1)}\footnote{#1: #2}}                                 \else
\nc{\commentbox}[1]{}
\newcommand{\mcomment}[2]{}
\newcommand{\acomment}[2]{}
\fi

\ifchanges

\newcommand{\loldnew}[3]{\commentbox{{\textcolor{blue}{\setlength{\fboxsep}{1pt}\fbox{\small
          #1}}} \textcolor{red}{\footnotesize #2}}
  \textcolor{blue}{#3}}
\setul{}{0.2mm}
\setstcolor{red}
\newcommand{\oldnew}[3]{{\textcolor{blue}{\setlength{\fboxsep}{1pt}\fbox{\small
        #1}}} \st{\footnotesize #2} \textcolor{blue}{#3}}

\else
\newcommand{\loldnew}[3]{#3}
\newcommand{\oldnew}[3]{#3}
\fi

\nc{\tzm}[1]{\mcomment{TZ}{#1}}
\nc{\tsm}[1]{\mcomment{TS}{#1}}
\nc{\tz}[1]{\acomment{TZ}{#1}}
\nc{\thz}[1]{\acomment{TZ}{#1}}
\nc{\ts}[1]{\acomment{TS}{#1}}

\nc{\tzon}[2][]{\oldnew{TZ}{#1}{#2}} 
\nc{\tson}[2][]{\oldnew{TS}{#1}{#2}}

\nc{\tzlon}[2][]{\loldnew{TZ}{#1}{#2}} 
\nc{\tslon}[2][]{\loldnew{TS}{#1}{#2}}

\renewcommand{\hat}[1]{\widehat{#1}}
\newcommand{\eqh}{\hat{=}}
\newcommand{\quotes}[1]{``#1''}

\newcommand{\columnWidth}{11cm}

\newcommand{\substruclemma}{Substructure Lemma\xspace}
\newcommand{\First}{\mtext{First}}
\newcommand{\List}{\mtext{List}}
\newcommand{\Last}{\mtext{Last}}
\newcommand{\In}{\mtext{In}}
\newcommand{\Out}{\mtext{Out}}
\newcommand{\Empty}{\mtext{Empty}}

\newcommand{\Odd}{\mtext{Odd}}
\newcommand{\odd}{\text{odd}}
\newcommand{\even}{\text{even}}

\newcommand{\Counter}{\mtext{Counter}}
\newcommand{\isEmpty}{\mtext{Empty}}
\newcommand{\Zero}{\mtext{Zero}}

\newcommand{\Succ}{\mtext{Succ}}
\newcommand{\Pred}{\mtext{Pred}}
\newcommand{\Max}{\mtext{Max}}
\newcommand{\numEdges}{\#\mtext{edges}}
\newcommand{\numNodes}{\#\mtext{nodes}}

\newcommand{\congruent}[2]{\sim_{#1, #2}}

\newcommand{\shortVersion}[1]{#1}
\newcommand{\longVersion}[2]{}

\newcommand{\apptheoremtitlefont}[1]{\textsc{#1}}
\newcommand{\apptheoremcontentfont}{\itshape}

\newcommand{\apponlystartmarker}{ $\blacktriangleright\blacktriangleright\blacktriangleright$ }
\newcommand{\apponlyendmarker}{ $\blacktriangleleft\blacktriangleleft\blacktriangleleft$ }
\newcommand{\apprepetitionstartmarker}{ $\blacktriangleright\blacktriangleright\blacktriangleright$ }
\newcommand{\apprepetitionendmarker}{ $\blacktriangleleft\blacktriangleleft\blacktriangleleft$ }

\newcommand{\initialAppendix}{
  \section*{Appendix}

  \setcounter{section}{1}
   \renewcommand{\thesection}{\Alph{section}}
   \counterwithin{theorem}{section}
   \counterwithin{definition}{section}
   \counterwithin{example}{section}

  In the appendix we give the proofs that have been ommitted in the main text. For proofs that are partially present in the main article, we repeat the full proof. Parts that are only repeated are marked by \apprepetitionstartmarker and \apprepetitionendmarker.  
  
}
  
\newcommand{\writeAppendix}{\initialAppendix}
\newcommand{\toAppendix}[1]{
  \makeatletter
   \g@addto@macro\writeAppendix{#1}
  \makeatother
}

\newcommand{\toMainAndAppendix}[1]{
      #1  \toAppendix{      \apprepetition{#1} \par
  }
}

\newcommand{\atheorem}[2]{
  \begin{theorem}\label{#1}    #2  \end{theorem}  \toAppendix{    \begin{apptheorem}{\ref{#1}}{}      #2    \end{apptheorem}  }
}

\newcommand{\alemma}[2]{
  \begin{lemma}\label{#1}    #2  \end{lemma}  \toAppendix{    \begin{applemma}{\ref{#1}}{}      #2    \end{applemma}  }
}

\newcommand{\aproof}[3]{  \@ifmtarg{#2}{}{    \begin{proof}      #1      #2    \end{proof}  }
  \toAppendix{    \begin{proof}      \@ifmtarg{#1}{}{\apprepetition{#1}} \par
      #3
    \end{proof}  }
}

\newcommand{\shortOrLong}[2]{
  \shortVersion{#1}
  \longVersion{#2}
}

\makeatletter
\newcommand{\theoremcont}[3]{
   \def\Type{#1}
   \def\Number{#2}
   \def\Label{#3}
  \@ifmtarg{#3}{
     \apptheoremtitlefont{\Type\ \Number.} \apptheoremcontentfont
   }{
    \apptheoremtitlefont{\Type\ \Number}\ \apptheoremcontentfont(\Label).
  }
}

\newenvironment{applemma}[2]{\vspace{2mm}\par\theoremcont{Lemma}{#1}{#2}}{\vspace{0mm}\par}
\newenvironment{apptheorem}[2]{\vspace{2mm}\par \theoremcont{Theorem}{#1}{#2}}{\vspace{2mm} \par }
\newenvironment{appcorollary}[2]{\theoremcont{Corollary}{#1}{#2}}{\vspace{2mm}}
\newenvironment{appproposition}[2]{\theoremcont{Proposition}{#1}{#2}}{\vspace{2mm}}
\newenvironment{appdefinition}[2]{\theoremcont{Definition}{#1}{#2}}{\vspace{2mm}}
\newenvironment{appexample}[1]{\vspace{2mm}\textit{Example #1.}}{\vspace{2mm}}

\newcommand{\apponlystart}{
    \apponlystartmarker
    }
\newcommand{\apponlyend}{
    \apponlyendmarker
    }

\newcommand{\apprepetition}[1]{
  \apprepetitionstartmarker #1 \apprepetitionendmarker
}

    \pgfdeclarelayer{background}
\pgfdeclarelayer{substructure}
\pgfdeclarelayer{edges}
\pgfdeclarelayer{foreground}
\pgfsetlayers{background,substructure,edges,main,foreground}

\tikzstyle{background rectangle}=[
  fill=black!5,
  draw=black!20,
  inner sep=0.5cm,
  rounded corners=5pt
]

\tikzstyle{class rectangle}=[
  draw=black,
  inner sep=0.2cm,
  rounded corners=5pt,
  thick
]

\tikzstyle{semantic rectangle}=[
  fill=blue!10, 
  draw=blue!30,
  inner sep=0.2cm,
  rounded corners=5pt
]

\tikzstyle{invisibleEdge}=[
  transparent
]
  
\tikzstyle{relnode}=[
   font=\normalsize,
   opaque
]

\tikzstyle{equal}=[
  blue!20,
  line width=1.8,
  preaction={line width=2.4,black,draw,shorten >=0.2,shorten <=0.2}
]

\tikzstyle{collapsematrix} = [
  matrix of nodes, 
  ampersand replacement=\&, 
  row sep=10, column sep=8,
]

\newcommand{\piccollapse}{
  \begin{tikzpicture}[
       xscale=1.2,
   ]
    \begin{scope}[shift={(0,0)}] \picabscollapse \end{scope}
    \begin{scope}[shift={(5.7,0)}] \picdeltacollapse \end{scope}

    \draw [class rectangle] (-3.1, 1) rectangle (9.1,-9.4);
    \draw [class rectangle, dashed] (-3.0, -0.6) rectangle (9.0,-9.3);
    \draw [class rectangle, dashed] (-2.9, -3.4) rectangle (8.9,-9.2);
    \draw [class rectangle, dashed] (-2.8, -4.6) rectangle (8.8,-9.1);
    \draw [class rectangle] (-2.7, -5.8) rectangle (8.7,-9.0);
    \draw [class rectangle] (-2.6, -7.8) rectangle (3.1,-8.9);
    \draw [class rectangle, dashed] (2.1, -7.7) rectangle (8.6,-8.8);

    \draw[invisibleEdge] (2.3,0) -- node[relnode] {\reflongeq{\ref{lemma:absdeltaequivalence}}}(3.0,0);
    \draw[invisibleEdge] (2.3,-2) -- node[relnode] {\reflongeq{\ref{theorem:dynefoandddynefoequivalence}}}(3.0,-2);
    \draw[invisibleEdge] (2.3,-6.5) -- node[relnode] {\reflongeq{\ref{lemma:absdeltaequivalence}}}(3.0,-6.5);

    \node[] (tmp) at (0, -9.7) {\textbf{Absolute Semantics}};
    \node[] (tmp) at (5.2, -9.7) {\textbf{$\Delta$-Semantics}};
    
    \begin{pgfonlayer}{substructure}
      \draw [semantic rectangle] (-2.5, 0.7) rectangle (2.5,-10.2);    
      \draw [semantic rectangle] (2.8, 0.7) rectangle (8.5,-10.2);    
    \end{pgfonlayer}
  \end{tikzpicture}

}

\newcommand{\picabscollapse}{
    \begin{scope}[shift={(0,0)}] \picabscollapsedynfo \end{scope}
    \begin{scope}[shift={(0,-2)}] \picabscollapsedynefo \end{scope}
    \begin{scope}[shift={(0,-4)}] \picabscollapsedynucq \end{scope}
    \begin{scope}[shift={(0,-5.2)}] \picabscollapsedynqf \end{scope}
    \begin{scope}[shift={(0,-6.7)}] \picabscollapsedynprop \end{scope}
    \begin{scope}[shift={(0,-8.4)}] \picabscollapsedynand \end{scope}

}

\newcommand{\picdeltacollapse}{
    \begin{scope}[shift={(0,0)}] \picdeltacollapsedynfo \end{scope}
    \begin{scope}[shift={(0,-2)}] \picdeltacollapsedynefo \end{scope}
     \begin{scope}[shift={(0,-6.6)}] \picdeltacollapsedynprop \end{scope}
     \begin{scope}[shift={(0,-8.2)}] \picdeltacollapsedynand \end{scope}
}

\newcommand{\picabscollapsedynfo}{
  \node[collapsematrix] (dynfo) at (0,0)
  {
    $\DynFO$ \& $\DynFOand$ \\
  };  
  \draw[invisibleEdge] (dynfo-1-1) -- node[relnode] {\refeq{\ref{theorem:dynfoequivalences}}}(dynfo-1-2);

}

\newcommand{\picabscollapsedynefo}{
  \node[collapsematrix] (dynefo) at (0,0)
  {
    $\DynEFO$ \& $\DynAFO$ \\
    $\textbf{\DynCQneg}$ \& $\DynUCQneg$ \\
  };  
    \draw[invisibleEdge] (dynefo-1-1) -- node[relnode] {\refeq{\ref{theorem:dynefoequivalences}}}(dynefo-1-2);
    \draw[invisibleEdge] (dynefo-2-1) -- node[relnode] {\refeq{\ref{theorem:dynefoequivalences}}}(dynefo-2-2);
    \draw[invisibleEdge] (dynefo-2-2) -- node[relnode, rotate=90] {$=$}(dynefo-1-2);
    \draw[invisibleEdge] (dynefo-1-1) -- node[relnode, rotate=90] {$=$}(dynefo-2-1);

}

\newcommand{\picabscollapsedynucq}{
  \matrix (dynucq) [collapsematrix]
  {
    $\textbf{\DynCQ}$ \& $\DynUCQ$ \\
  }; 
    \draw[invisibleEdge] (dynucq-1-1) -- node[relnode] {\refeq{\ref{theorem:dynucqequivalences}}}(dynucq-1-2);
}

\newcommand{\picabscollapsedynqf}{
  \matrix (dynucq) [collapsematrix]
  {
    $\DynQF$ \\
  };

}

\newcommand{\picabscollapsedynprop}{
  \matrix (dynprop) [collapsematrix]
  {
    $\DynProp$ \&  \DynPropUCQneg\\
    $\textbf{\DynPropCQneg}$ \& \DynPropUCQ \\
  };  
    \draw[invisibleEdge] (dynprop-1-1) -- node[relnode] {$=$}(dynprop-1-2);
    \draw[invisibleEdge] (dynprop-1-1) -- node[relnode, rotate=90] {\refeq{\ref{theorem:dynpropequivalences}}}(dynprop-2-1);
    \draw[invisibleEdge] (dynprop-2-2) -- node[relnode, rotate=90] {\refeq{\ref{theorem:dynpropequivalences}}}(dynprop-1-2);
    \draw[invisibleEdge] (dynprop-2-2) -- node[relnode] {$=$}(dynprop-2-1);

}

\newcommand{\picabscollapsedynand}{
  \matrix (dynand) [collapsematrix]
  {
    \textbf{$\DynAnd$}\\
  };  
}

\newcommand{\picdeltacollapsedynfo}{
  \node[collapsematrix] (ddynfo) at (0,0)
  {
    $\dDynFO$ \& $\dDynFOpos$ \\
  };  
    \draw[invisibleEdge] (ddynfo-1-1) -- node[relnode] {\refeq{\ref{lemma:deltanegationfree}}}(ddynfo-1-2);
}

\newcommand{\picdeltacollapsedynefo}{
  \node[collapsematrix] (ddynefo) at (0,0)
  {
    $\dDynEFO$ \& $\dDynAFO$ \\
    $\dDynCQneg$ \& $\dDynUCQneg$ \\
    \textbf{$\dDynCQ$} \& $\dDynUCQ$ \\
  };  
    \draw[invisibleEdge] (ddynefo-1-1) -- node[relnode] {\refeq{\ref{theorem:ddynefoequivalences}}}(ddynefo-1-2);
    \draw[invisibleEdge] (ddynefo-2-1) -- node[relnode] {\refeq{\ref{theorem:ddynefoequivalences}}}(ddynefo-2-2);
    \draw[invisibleEdge] (ddynefo-2-2) -- node[relnode, rotate=90] {$=$}(ddynefo-1-2);
    \draw[invisibleEdge] (ddynefo-1-1) -- node[relnode, rotate=90] {$=$}(ddynefo-2-1);
    \draw[invisibleEdge] (ddynefo-3-1) -- node[relnode] {\refeq{\ref{theorem:ddynefoequivalences}}}(ddynefo-3-2);
    \draw[invisibleEdge] (ddynefo-3-1) -- node[relnode, rotate=90] {\refeq{\ref{lemma:deltanegationfree}}}(ddynefo-2-1);
    \draw[invisibleEdge] (ddynefo-3-2) -- node[relnode, rotate=90] {\refeq{\ref{lemma:deltanegationfree}}}(ddynefo-2-2);
}

\newcommand{\picdeltacollapsedynprop}{
  \matrix (ddynprop) [collapsematrix]
  {
    $\dDynProp$ \& \dDynPropUCQneg  \\
    $ $ \& $ $ \\
  };  
    \node[] (tmp) at (-0.5, -0.6) {\textbf{\dDynPropUCQ}};
    \draw[invisibleEdge] (ddynprop-1-1) -- node[relnode] {$=$}(ddynprop-1-2);
    \draw[invisibleEdge] (tmp) -- node[relnode, rotate=90] {\refeq{\ref{lemma:deltanegationfree}}}(ddynprop-1-1);
    \draw[invisibleEdge] (tmp) -- node[relnode, rotate=90] {\refeq{\ref{lemma:deltanegationfree}}}(ddynprop-1-2);

}

\newcommand{\picdeltacollapsedynand}{
  \matrix (ddynprop) [collapsematrix]
  {
    $\textbf{\dDynPropCQ}$ \& \dDynPropCQneg  \\
  };  
    \draw[invisibleEdge] (ddynprop-1-1) -- node[relnode] {\refeq{\ref{lemma:deltanegationfree}}}(ddynprop-1-2);

}

   \author{Thomas Zeume\thanks{\texttt{thomas.zeume@cs.tu-dortmund.de}}$\;$}      
   \author{Thomas Schwentick\thanks{\texttt{thomas.schwentick@tu-dortmund.de}}}
    
   \affil{TU Dortmund University}
   \title{Dynamic Conjunctive Queries\footnote{An extended abstract of this work appeared in the proceedings of the conference International Conference on Database Theory 2014 (ICDT 2014) \cite{ZeumeS14icdt}. Some preliminary results appeared in the proceedings of the conference Mathematical Foundations of Computer Science 2013 (MFCS 2013)\cite{ZeumeS13}. Both authors acknowledge the financial support by DFG grant SCHW 678/6-1.}}

  \maketitle
  \begin{abstract}

    The article investigates classes of queries maintainable by 
conjunctive queries ($\CQ$s) and
their extensions and restrictions in the dynamic complexity framework of
Patnaik and Immerman.  Starting from the basic language of quantifier-free conjunctions of
positive atoms, it studies the impact of additional operators and
features --- such as union, atomic negation
and quantification --- on the dynamic expressiveness, for the
standard semantics as well as for $\Delta$-semantics. 

Although  many different combinations of these
features are possible,  they basically yield five important fragments for the
standard semantics,
characterized by the addition of  (1) arbitrary quantification and atomic negation, (2) existential
quantification and atomic negation, (3) existential quantification,
(4) atomic negation (but no quantification)), and by (5) no addition to the basic language at all. While fragments
(3), (4) and (5) can be separated, it remains open whether fragments (1), (2) and (3) 
are actually different. The fragments arising from
$\Delta$-semantics are also subsumed by the standard fragments (1),
(2) and (4). The main fragments of \DynFO that had been studied in
previous work, \DynQF and \DynProp, characterized by quantifier-free
update programs with or without auxiliary functions, respectively,
also fit into this hierarchy: \DynProp coincides with fragment (4) and
\DynQF is strictly above fragment (4) and within fragment
(3).

As a further result, all (statically) \FO-definable queries are
captured by fragment (2) and a complete characterization of these
queries in terms of non-recursive dynamic  $\QFO[\exists^1]$-programs is given.
\end{abstract}

\tableofcontents

\sloppy

  \section{Introduction}\label{section:intro}
    The re-evaluation of a fixed query after a modification to a huge database can be a time-consuming process; in particular when it is performed from scratch. For this reason previously computed information such as the old query result and (possibly) other auxiliary information is often reused in order to speed up the process.
This maintenance of query results has attracted lots of attention over the last decades of database related research, in particular in the related field of view maintenance. For relational databases algorithmic (see e.g.~\cite{ShmueliI84, GuptaMS1993}) and declarative approaches (see e.g.~\cite{DongT92, DongS93, PatnaikI94}) have been studied. 

Here, we continue the study of the declarative approach where query results are maintained by queries from some query language. More precisely, for a relational data\-base  subject to change, auxiliary relations are maintained with the intention to help answering a query $\query$. When a modification to the database, i.e.~an insertion or deletion of a tuple, occurs, every auxiliary relation is updated through a first-order query that can refer to both, the database and the auxiliary relations. 

One possible formalization of this approach is the descriptive dynamic complexity framework (short: dynamic complexity) by Patnaik and Immerman \cite{PatnaikI94}. In their framework, the class $\DynFO$ contains of all queries maintainable through first-order updates (and thus also in the core of SQL). This is the formalization used in this work.

Shortly before the work of Patnaik and Immerman, the declarative approach was independently formalized in a very similar way by Dong, Su and Topor \cite{DongT92, DongS93}. For a discussion of the differences of both formalizations we refer to the later discussion of the  choice of the precise setting. 

The class $\DynFO$ is quite powerful. Many queries inexpressible in (static) first-order logic, such as the transitive closure query on undirected graphs \cite{PatnaikI94} and the word problem for context-free languages \cite{GeladeMS12}, can be maintained in \DynFO. 
There are no general inexpressibility results for $\DynFO$ at all\footnote{Except for the trivial ones due to the fact that queries maintainable in \DynFO can be computed in polynomial time.}. 

Towards a deeper understanding of the dynamic maintainability of queries, two main restrictions of $\DynFO$ have been explored in the literature. Dong and Su started the study of restricted auxiliary relations \cite{DongS98}  and obtained inexpressibility results for unary auxiliary relations. On the other hand, Hesse started the exploration of syntactic fragments of $\DynFO$, such as the one obtained by disallowing quantification in update formulas \cite{Hesse03}. Inexpressibility results for this particular fragment have been obtained \mbox{in \cite{GeladeMS12}}.

In this work, we investigate classes of queries maintainable by conjunctive queries and extensions thereof, thus we are following the approach of Hesse.

Conjunctive queries ($\CQ$s), that is, in terms of logic, existential first-order queries whose quantifier-free part is a conjunction of atoms, are one of the most investigated query languages. Starting with Chandra and Merlin \cite{ChandraM1977}, who analyzed conjunctive queries for relational databases, those queries have been studied for almost every emerging new database model. Usually also the extension by unions ($\UCQ$s), by negations ($\CQneg$s) as well as by both unions and negations ($\UCQneg$s or, equivalently, $\EFO$) have been studied. It is folklore that all those classes are distinct for relational databases.

In this work we aim at the following goals. 

\begin{goal}
  Understand the relative expressiveness of different extensions and restrictions of dynamic conjunctive queries in the dynamic setting, as well as their ability  to maintain queries from (stronger) static query classes.
\end{goal}

For extensions, we add negation and/or disjunction to conjunctive queries, and for restrictions we disallow quantification. We further also consider universal quantification in place of existential quantification.

As for the relationship to static classes, it is interesting to understand whether larger static classes $\calC$ can be captured by  dynamic classes $\DynC'$, for weaker $\calC'$. Up to now only two such results were known, namely, that $\MSO$ can be characterized by quantifier-free $\DynFO$ on strings and that, on general structures, $\EFO$ is captured by \DynQF,  the quantifier-free fragment of $\DynFO$ with auxiliary relations \emph{and} auxiliary functions \cite{GeladeMS12}.

In both dynamic complexity frameworks auxiliary relations are always explicitly defined as a whole after a modification. 
However, in the context of query re-evaluation, it is often convenient to express the new state of an auxiliary relation $R$ in terms of the current relation and some ``Delta'', that is,  by specifying a tuple set $R^+$ to be added to $R$ and a tuple set $R^-$ to be removed from $R$. We refer to the former semantics as \emph{absolute semantics} and to the latter as \emph{$\Delta$-semantics}. Obviously, the choice of the semantics does not affect the expressiveness of an update language that is closed under Boolean operations. However, most of the update languages in this paper lack some Boolean closure properties.

\begin{goal}
  Understand the relationship between absolute semantics and $\Delta$-semantics for conjunctive queries and their variants.
\end{goal}

In this work we contribute to achieve those two goals as follows.

\paragraph{Contributions.}
 For an overview of the relationship of the various dynamic classes of conjunctive queries we refer to \mbox{Figure \ref{figure:collapse}}.

  \begin{figure*}
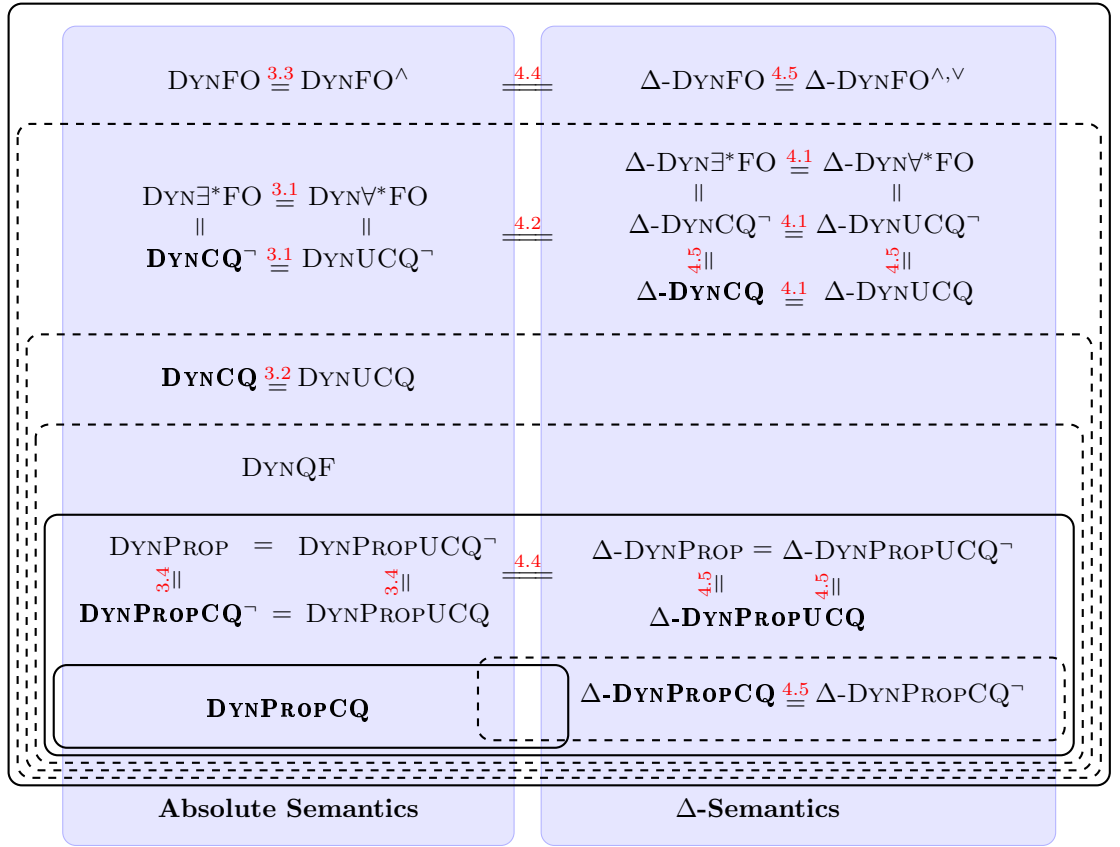

    \begin{center}
       \scalebox{1}{
  \piccollapse
       }
      \caption{Hierarchy of fragments of \DynFO. Solid lines are strict separations. \label{figure:collapse} }
    \end{center}
  \end{figure*}

The distinctness of the underlying static query classes does not translate into the dynamic setting:\footnote{The notation for classes will be formally introduced in Sections \ref{section:dynamicsetting} and \ref{section:deltasemantics}. In general, $\Delta$ indicates $\Delta$-semantics, the absence of $\Delta$ indicates absolute semantics.}
\begin{itemize}
  \item  We show that, in many cases, the addition of the union-operator does not yield additional expressive power in the dynamic setting, for example, $\DynUCQneg = \DynCQneg$,  $\DynUCQ = \DynCQ$, and $\DynPropUCQneg=\DynPropCQneg$, where $\Prop$ indicates classes without quantifiers.
  \item Furthermore, negation often does not increase the expressive power of an update language, e.g.~we have \mbox{$\DynPropUCQneg=\DynPropUCQ$} and $\dDynCQneg = \dDynCQ$.  
  \item Finally, often quantifiers can be replaced by their dual quantifiers, e.g.~$\DynEFO (= \DynUCQneg) = \DynAFO$.
\end{itemize}

Whether $\DynCQneg = \DynCQ$ remains open. However, first-steps towards the separation of the remaining fragments have been taken. Already  in \cite{ZeumeS15reach} the authors proved that dynamic conjunctive queries without negations and quantifiers are strictly weaker than the quantifier-free fragment of $\DynFO$. Here we continue this work:
\begin{itemize}
  \item By showing that $\DynQF$, the extension of $\DynProp$ that allows auxiliary functions, is contained in $\DynCQ$, we can separate the classes  $\DynProp$ and $\DynCQ$.
\end{itemize}

Furthermore, we show that dynamic conjunctive queries extended by negations capture all first-order queries:
\begin{itemize}
 \item We characterize the class of first-order queries as the class maintainable by non-recursive dynamic $\DynEFO$-programs with a single existential quantifier per update formula. This implies that dynamic conjunctive queries extended by negations can maintain all first-order queries.
\end{itemize}

For the second goal, the main finding is that the difference between absolute and $\Delta$-semantics is much smaller than we had expected.
\begin{itemize}
\item The dynamic classes corresponding to $\FO$, $\CQneg$ and $\Prop$ yield the same expressive power with respect to absolute and $\Delta$-semantics.
\item It turns out that conjunctive queries and conjunctive queries with negation coincide with respect to $\Delta$-semantics, that is, in particular, $\dDynCQ=\dDynCQneg$ and thus, also $\dDynCQ=\DynUCQneg$ . 
\end{itemize}

\paragraph{Choice of setting.}
The concrete settings under which dynamic complexity has been studied in the literature slightly differ in several aspects. We shortly discuss the most important aspects, what choice we took for this work and why we made this choice.

An important aspect is whether to use a finite and fixed domain, an active domain or an infinite domain. In this work, we follow the framework of Patnaik and Immerman in which the domain is finite and fixed \cite{PatnaikI94}. To maintain a query, a dynamic program has to work uniformly for all domains. This fixed domain framework for a dynamic setting might appear counterintuitive at first sight. However, it allows to study the underlying dynamic mechanisms of dynamic programs, in particular when one is interested to develop  lower bound methods. Fixed domains also offer a strong connection to logics and circuit complexity. In incremental evaluation systems (IES), a framework proposed by Dong and Topor \cite{DongT92}, active domains are used. First-order incremental evaluation systems (FOIES), introduced by Dong and Su in subsequent work \cite{DongS93}, have a very close connection to $\DynFO$.   This setting is a little closer to real database systems but most results in dynamic complexity hold equally in 
both frameworks.

Another parameter to choose is how the auxiliary data is initialized. In the setting of Patnaik and Immerman, dynamic programs start from empty databases and the auxiliary data is either initialized by a polynomial time computation or by a formula from the same class as the update formulas. Later this was generalized by Weber and the second author by proposing that dynamic programs start from an arbitrary initial database and auxiliary data initialized by a mapping computable in some given complexity \mbox{class \cite{WeberS07}}.

In the present work, we allow for arbitrary initialization mappings. This is motivated by our long term goal to develop lower bound techniques for dynamic programs. While lower bounds in settings with restricted initialization might depend on this restriction, an inexpressibility result in the setting with arbitrary initialization, on the other hand, really shows that a query cannot be \emph{maintained}. A result like $\DynUCQ = \DynCQ$ is helpful for the development of lower bound techniques, as it shows that for proving lower bounds for $\DynUCQ$ it is sufficient to consider $\DynCQ$ programs --- but also that one has to be aware that lower bounds for  $\DynCQ$  are as hard as lower bounds for  $\DynUCQ$.
However, though all our results are stated for arbitrary initialization mappings, they also hold in the setting with empty initial database and first-order initialization for the auxiliary data. On the other hand, some proofs do not carry over to the strict setting of Patnaik and Immerman where, in a dynamic class $\DynC$, only $\class$ initializations are allowed.

\paragraph{Related work.}
We next discuss some further related work, beyond what we already mentioned above. The expressivity of first-order logic in the dynamic complexity frameworks discussed above has been studied a lot (see e.g.~\cite{PatnaikI94, DongS98, Etessami98, Hesse03b, Hesse03, WeberS07, GraedelS12, GeladeMS12}). Most results focus on showing that a problem from some static complexity class can be dynamically maintained by programs of a weaker query class. Some lower bounds have been achieved as well (see e.g.~\cite{DongLW95, DongLW03, DongS98, GeladeMS12, GraedelS12, ZeumeS15reach}). Many other aspects such as the arity of auxiliary relations (see e.g.~\cite{DongS98, Hesse03}), whether the auxiliary relations are determined by the current structure (see e.g.~\cite{PatnaikI94, DongS97, GraedelS12}),  and the presence of an order (see e.g.~\cite{GraedelS12})   have been studied. 

An algebraic perspective of incremental view maintenance under $\Delta$-semantics has been studied in \cite{Koch10}. Parts of the latter work have also been implemented, see e.g.~\cite{KochAKNNLS14}.

\paragraph{Outline.}
In Section \ref{section:dynamicsetting} we define our dynamic setting more precisely. In \mbox{Section \ref{section:fragments}} the concrete dynamic complexity classes under consideration are introduced. Results for the collapse of classes are obtained in Section \ref{section:collapse} and results for the separation classes are proved in Section \ref{section:separations}. The alternative $\Delta$-semantics is introduced and studied in Section \ref{section:deltasemantics}. In \mbox{Section \ref{section:dyncharact}} we give the dynamic characterization of first-order logic. We conclude with a discussion and a first step toward separations in Section \ref{section:conclusion}. 

\paragraph{Acknowledgement.}
We thank Nils Vortmeier for careful proofreading. Further we are grateful to the anonymous reviewers of this as well as preceeding work for several very helpful comments.

  \section{Dynamic setting}\label{section:dynamicsetting}
    In this section, we introduce the basic concepts and fix our notation. We mainly borrow it from our previous \mbox{work \cite{ZeumeS15reach}}. 

A \emph{dynamic instance} of a query $\query$ is a pair $(\db, \alpha)$, where $\db$ is a database over a finite domain $\domain$ and $\alpha$ is a sequence of modifications to $\db$, i.e.~a sequence of insertions and deletions of tuples over $\domain$. The dynamic query $\dynProb{$\query$}$ yields as result the relation that is obtained by first applying the modifications \mbox{from $\alpha$} to $\db$ and then evaluating the query $\query$ on the resulting  database. 

The database resulting from applying a modification $\delta$ to a database $\db$ is denoted by $\delta(\db)$. The result $\updateDB{\alpha}{\db}$ of applying a sequence of modifications $\alpha = \delta_1 \ldots \delta_m$ to a database $\db$ is defined by $\updateDB{\alpha}{\db} \df \updateDB{\delta_m}{\ldots (\updateDB{\delta_1}{\db})\ldots}$.

Dynamic programs, to be defined next, consist of an initialization mechanism and an update program.  The former  yields, for every (input) database $\db$,  an initial state with initial auxiliary  data. The latter defines the new state of the dynamic program for each possible modification $\delta$.

A \emph{dynamic schema} is a tuple\footnote{In \cite{ZeumeS15reach} a dynamic schema had an additional schema for an extra database with built-in relations. As here we do not restrict auxiliary relations in any way and allow arbitrary initialization,  we do not need built-in relations.} \mbox{$(\inpSchema, \auxSchema)$} where $\inpSchema$ and $\auxSchema$ are the schemas of the input database and the auxiliary database, respectively. In this work all schemata are purely relational, although all results also hold for input schemas with constants. We always let $\tau\df\inpSchema\cup\auxSchema$. 

\begin{definition}(Update program)\label{def:updateprog}
  An \emph{update program} \prog over dynamic schema \mbox{$(\inpSchema, \auxSchema)$} 
  is a set of first-order formulas (called \textit{update formulas} in the following) that contains,  for every $R \in \auxSchema$ and every
  $\delta \in \{\ins_S, \del_S\}$ with $S \in \inpSchema$, an update formula  $\uf{R}{\delta}{\vec x}{\vec y}$ over the schema $\schema$  where $\vec x$ and $\vec y$ have the same arity as $S$ and $R$, respectively.
\end{definition}

A \emph{program state} $\state$ over dynamic schema \mbox{$(\inpSchema, \auxSchema)$} is a structure $(\domain, \inp,  \aux)$ where $D$ is a finite domain, $\inp$ is a database over the input schema (the \emph{current database}) and $\aux$ is a database over the auxiliary schema (the \emph{auxiliary database}).
The \emph{semantics of update programs} is as follows. For a modification $\delta(\vec a)$, where $\vec a$ is a tuple over $\domain$,  and program state $\state=(\domain, \inp,\aux)$ we denote by $P_\delta(\state)$ the state $(\domain, \delta(\inp), \aux')$, where $\aux'$ consists of relations \mbox{$R'\df\{\vec b \mid \state \models \uf{R}{\delta}{\vec a}{\vec b}\}$}. The effect $P_\alpha(\state)$ of a modification sequence $\alpha = \delta_1 \ldots \delta_m$ to a state $\state$ is the state $\updateState{P}{\delta_m}{\ldots (\updateState{\prog}{\delta_1}{\state})\ldots}$. 

\begin{definition}(Dynamic program) \label{definition:dynprog}
  A \emph{dynamic program} is a triple $(P,\init,Q)$, where
  \begin{compactitem}
   \item  $P$ is an update program over some dynamic schema
  \mbox{$(\inpSchema, \auxSchema)$}, 
    \item \init is a mapping that maps $\inpSchema$-databases to $\auxSchema$-databases, and 
    \item $Q\in\auxSchema$ is a designated \emph{query symbol}.
  \end{compactitem}
\end{definition}

A dynamic program $\calP=(P,\init,Q)$ \emph{maintains}  a dynamic query  $\dynProb{$\query$}$ if, for every dynamic instance $(\db,\alpha)$, the relation $\query(\alpha(\db))$ coincides with the query relation $Q^\state$ in the state \mbox{$\state=P_\alpha(\state_\init(\db))$}, where $\state_\init(\db)$ is the initial state, i.e.~\mbox{$\state_\init(\db) \df (\domain, \db,  \auxInit(\db))$}.

Several dynamic settings and restrictions of dynamic programs have been studied in the literature (see e.g.~\cite{PatnaikI94, Etessami98, GraedelS12, GeladeMS12}). Possible parameters are, for instance
\begin{compactitem}
\item  the logic in which update formulas are expressed;
\item whether, in dynamic instances $(\db,\alpha)$, the initial data\-base $\db$ is always empty; and
\item whether the initialization mapping $\auxInit$ is \emph{permu\-tation-invariant} (short: \textit{invariant}), that is, whether $\pi(\auxInit(\db))=\auxInit(\pi(\db))$ holds, for every data\-base $\db$ and permutation $\pi$ of the domain.
\end{compactitem}
We refer to the introduction for a discussion of the choices made in the following definition.
\begin{definition}(\DynC) \label{definition:dync}
  For a class $\calC$ of formulas, let $\DynC$ be the class of all dynamic queries that  can be maintained by dynamic
  programs with formulas from $\calC$ and arbitrary initialization mapping. 
\end{definition}
In particular $\DynFO$ is the class of all dynamic queries that  can be maintained by first-order update formulas. $\DynProp$ is the subclass of $\DynFO$, where update formulas are not allowed to use quantifiers.

We note that arbitrary, (possibly) non-uniform initialization mappings permit to maintain undecidable queries, even when the logic for expressing update formulas is very weak. 

Allowing arbitrary initialization mappings in Definition \ref{definition:dync} helps us to concentrate on the \emph{maintenance} aspect of dynamic complexity and it helps keeping proofs short. All our results also hold for $\FO$-definable initialization mappings on ordered domains. 
\begin{example}\label{example:acyclictc}
  The transitive closure of an acyclic graph can be maintained in $\DynFO$.  We follow the argument from \cite{PatnaikI94} and construct a dynamic $\DynFO$-program with one binary auxiliary relation $T$  which is intended to store the transitive closure.
  
  Insertions can be handled straightforwardly. After inserting an edge $(u,v)$ there is a path from $x$ to $y$ if, before the insertion, there has been a path from $x$ to $y$ or there have been paths from $x$ to $u$ and from $v$ to $y$. There is a path $p$ from $x$ to $y$ after deleting an edge $(u,v)$ if there was a path from $x$ to $y$ before the deletion (1) but there was no such path via $(u,v)$, or (2) there is an edge $(z,z')$ on $p$ such that $u$ can be reached from $z$ but not from $z'$. If there is still a path $p$ from $x$ to $y$, such an edge $(z,z')$ must exist on $p$, as otherwise $u$ would be reachable from $y$ contradicting acyclicity. All conditions  can be checked using the transitive closure of the graph before the deletion of $(u, v)$. The update formulas for $T$ are as follows:

 \begin{align*}
    \uf{T}{\ins E}{u,v}{x,y} \df & \; T(x,y) \vee \big(T(x,u) \wedge T(v,y)\big) \\
    \uf{T}{\del E}{u,v}{x,y} \df & \; T(x,y) \wedge \Big( \big(\neg T(x,u) \vee \neg T(v,y)\big) \\ 
     & \quad \vee \exists z \exists z' \big(T(x,z) \wedge E(z,z') \wedge (z \neq u \vee z' \neq v) \\ 
     & \quad \quad \quad \quad \quad \wedge T(z',y) \wedge T(z,u) \wedge \neg T(z',u)\big)\Big)
 \end{align*}
 \qed
\end{example}

The following notion will be useful at several occasions. The \emph{dependency graph} of a dynamic program $\prog$ with auxiliary schema $\schema$ has vertex set $V = \schema$ and an edge $(R, R')$ if the relation symbol $R'$ occurs in one of the update formulas for $R$. 
The \emph{deletion dependency graph} is defined as the dependency graph except that only update formulas for delete operations are taken into account.

  \section{Dynamic conjunctive queries}\label{section:fragments}
    
In this section we study the relationship between the variants of dynamic conjunctive queries. 

We first give formal definitions of the classes of queries we are interested in:
\begin{itemize}
  \item $\CQ$: the class of conjunctive queries, that is, queries
  expressible by first-order formulas of the form $\varphi(\vec x) =
  \exists \vec y \psi$, where $\psi$ is a conjunction of atomic formulas.
  \item $\UCQ$: the class of all unions of conjunctive
  queries, that is, queries expressible by formulas of the form
  $\bigvee_i \exists \vec x \psi_i$, where each $\psi_i$ is a conjunction of atomic formulas.
\end{itemize}
We note that safety of queries is not an issue in this paper: we use queries as update formulas only and we can always assume that, for each required arity,  there is an auxiliary ``universal'' relation containing all tuples of this arity over the active domain which could be used to make queries syntactically safe. 

The classes $\CQ$ and $\UCQ$ can be extended by additionally allowing negated atoms, resulting in $\CQneg$ and $\UCQneg$; or they can be restricted by disallowing quantification, resulting in $\PropCQ$, $\PropUCQ$, $\PropCQneg$ and $\PropUCQneg$.  It is well known that $\UCQneg$ and $\EFO$, the class of queries expressible by existential first-order formulas, coincide, but otherwise,  all these classes are distinct. Furthermore, other quantification patterns than $\exists^*$ can be considered, like $\forall^*$ or arbitrary quantification. 

The dynamic program maintaining the transitive closure for acyclic graphs in Example \ref{example:acyclictc} is actually a $\DynUCQneg$-program.

The main goal of this section is to show that the relationship of all these classes in the dynamic setting is much simpler than in the static setting. In the first part of the section we prove that many dynamic classes collapse, as indicated in the left part of Figure~\ref{figure:collapse}. In the second part we show that the dynamic classes $\DynPropCQ$, $\DynProp$ and $\DynCQ$ can be separated. Whether $\DynCQ$, $\DynUCQ$ and $\DynFO$ can be separated remains open.

    \subsection{Collapse results}\label{section:collapse}
      In this section we prove that dynamic classes collapse as indicated in the left part of Figure~\ref{figure:collapse}. More precisely, we show the following theorems. The main results of this section are the following two theorems regarding the second and the third fragment in the left part of Figure~\ref{figure:collapse}.

\begin{theorem}\label{theorem:dynefoequivalences}
  Let $\query$ be a query. Then the following statements are equivalent:
  \begin{enumerate}
    \item $\query$ can be maintained in $\DynUCQneg$.
    \item $\query$ can be maintained in $\DynCQneg$.
    \item $\query$ can be maintained in $\DynEFO$.
    \item $\query$ can be maintained in $\DynAFO$.
  \end{enumerate}
\end{theorem}

\begin{theorem}\label{theorem:dynucqequivalences}
  Let $\query$ be a query. Then the following statements are equivalent:
  \begin{enumerate}
    \item $\query$ can be maintained in $\DynUCQ$.
    \item $\query$ can be maintained in $\DynCQ$.
  \end{enumerate}
\end{theorem}

Using the same technique as is used for removing unions from dynamic unions of conjunctive queries, a normal form for $\DynFO$ can be obtained. The class $\DynFOand$ contains all queries maintainable by a program whose update formulas are in prenex normal form where the quantifier-free part is a conjunction of atoms.

\begin{theorem}\label{theorem:dynfoequivalences}
  Let $\query$ be a query. Then the following statements are equivalent:
  \begin{enumerate}
    \item $\query$ can be maintained in $\DynFO$.
    \item $\query$ can be maintained in $\DynFOand$.
  \end{enumerate}
\end{theorem}

Further we prove the following result for the quantifier-free variants of dynamic conjunctive queries.

\begin{theorem}\label{theorem:dynpropequivalences}
  Let $\query$ be a query. Then the following statements are equivalent:
  \begin{enumerate}
    \item $\query$ can be maintained in $\DynProp$.
    \item $\query$ can be maintained in $\DynPropUCQneg$.
    \item $\query$ can be maintained in $\DynPropCQneg$.
    \item $\query$ can be maintained in $\DynUCQ$.
  \end{enumerate}
\end{theorem}

Before we turn to the proofs of the theorems, we discuss the proof techniques that will be used.

For showing that a class $\DynC$ of queries is contained in a class $\DynC'$, it is sufficient to construct, for every dynamic program with update queries from class $\calC$, an equivalent dynamic program with update queries from class $\calC'$. In cases where  $\calC' \subset \calC$ this can also be seen as constructing a $\calC'$-normal form for $\calC$-programs. 

Most of the proofs for the collapse of two dynamic classes in this paper are not very deep. Indeed, most of them use one or more of the following three (easy) techniques.

The \emph{replacement technique} (\cite{ZeumeS15reach}) is used to remove subformulas of a certain kind from update formulas and to replace their ``meaning'' by additional auxiliary relations. In this way, we often can remove negations (choose negative literals as subformulas, see the proof of Lemma \ref{lemma:negationfree}) and disjunctions (see proof of Lemma \ref{lemma:disjunctionfree}) from update formulas. 

The \emph{preprocessing technique} is used to convert (more) complicated update formulas into easier update formulas by splitting the computation performed by the complicated update formula into two parts; one of them performed by the initialization mapping and stored in an additional auxiliary relation, the other one performed by the easier update formula using the pre-computed auxiliary relation. Applications of this technique are the removal of unions from dynamic unions of conjunctive queries (see example below) as well as proving the equivalence of semantics for dynamic conjunctive queries with negations (see \mbox{Lemma \ref{lemma:qfotodeltaqfo}}). 

\begin{example}\label{example:disjunctionfree}
  We consider the update formula $$\uf{R}{\delta}{u}{x} = \exists y \big(U(x,y) \lor V(x,u) \big)$$ for a unary relational symbol $R$. We aim at an equivalent update formula $\ufb{R}{\delta}{u}{x}$ without disjunction. The idea is to store a `disjunction blue print' in a precomputed auxiliary relation $T$ and to use existential quantification to guess which disjunct becomes true. 

  In this example, we assume that in every state of the dynamic program on  every database, both the interpretations of $U$ and $V$ are always non-empty sets.\footnote{This assumption will eventually be removed in the proof.} 
Then, $\uf{R}{\delta}{u}{x}$ can be replaced by 
      $$\exists y \exists z_1 \exists z_2  \exists z_3 \exists z_4 \big ( U(z_1,z_2) \land V(z_3,z_4) \land T(z_1, z_2, z_3, z_4, x,y,u)\big)$$
 
  where $T$ is an additional auxiliary relation symbol which is interpreted, in every state $\state$, by a $7$-ary relation $T^\state$ containing all tuples $(a_1, \ldots, a_7)$ with $(a_1, a_2) = (a_5,a_6)$ or $(a_3, a_4) = (a_5,a_7)$.  Thus $T^\state$ ensures that either the values chosen for $z_1, z_2$ coincide with the values of $x, y$ or the values of $z_3, z_4$ coincide with $x, u$. 

  Therefore, the initialization mapping initializes $T$  with the result of the query
        $$\query_T(z_1, z_2, z_3, z_4, x,y,u) \df \big((z_1, z_2) = (x, y) \lor (z_3, z_4) = (x, u)\big).$$

  Observe that this approach fails when $U$ or $V$ are interpreted by empty relations. In order to cover empty relations as well, some extra work needs to be done (see proof of Lemma \ref{lemma:disjunctionfree}).

\end{example}

The \emph{squirrel technique} maintains additional auxiliary relations that reflect the state of some auxiliary relation after every possible single modification (or short modification sequence).\footnote{Squirrels usually make provisions for every possible future.} For example, if a dynamic program contains a relation symbol $R$ then a fresh relation symbol $R_\ins$ can be used, such that the interpretation of $R_\ins$ contains the content of $R$ after modification $\ins$ (for every possible insertion tuple). Of course, $R_\ins$ has higher arity than $R$, as it takes the actual inserted tuple into account.
Sample applications of this technique are the removal of quantifiers from \emph{some} update formulas (see the following example and Lemma \ref{lemma:shifting}) and the maintenance of first-order queries in $\DynCQneg$ (see \mbox{Theorem \ref{theorem:non-recursive})}.

\begin{example}\label{example:shifting}
  Consider the update formula 
  $$\uf{Q}{\ins}{u_1}{x} = \exists y \big(Q(x) \vee \neg S(u_1, y) \big)$$
  for the query symbol $Q$ of some dynamic program $\prog$. In order to obtain a quantifier-free update formula for $Q$ after insertion of an arbitrary tuple we maintain the relation $Q_\ins(\cdotp, \cdotp)$ that contains a tuple $(a, b)$ if and only if $b$ would be in $Q$ in the next state, after insertion of $a$. Similarly for $S$ and deletions.

  Then the update formula $\ufwa{Q}{\ins}$ can be replaced by the quantifier-free formula $\uf{Q}{\ins}{u_1}{x} \df Q_\ins(u_1,x)$. The relation $Q_\ins$ can be updated via
    \begin{align*}
    \uf{Q_\ins}{\ins}{u_0}{u_1, x}& \df \exists y \big(Q_\ins(u_0, x) \vee \neg S_\ins(u_0, u_1, y) \big) \\
    \uf{Q_\ins}{\del}{u_0}{u_1, x}& \df \exists y \big(Q_\del(u_0, x) \vee \neg S_\del(u_0, u_1, y) \big)
    \end{align*}
and similarly, for the other new auxiliary relations.
\end{example}

We note that, in this example, the application of the technique does not eliminate all quantifiers in the program (in fact, it removes one and introduces two new formulas with quantifiers), but it removes quantification from the update formula for a \emph{particular relation}. Removing  quantification from the update formulas of the query relation will turn out to be useful in the proofs of Lemmata \ref{lemma:disjunctionfree}, \ref{lemma:quantifierswitch} and \ref{lemma:qfotodeltaqfo}.

This concludes the description of the techniques. In the following, as a preparatory step, we show how to remove quantification from the update formulas of the query relation. Afterwards we exhibit constructions for removing negations and unions from some dynamic fragments. Finally we present a construction for switching quantifiers in dynamic programs, that is, e.g., for constructing an $\EFO$-program from an $\AFO$-program.  

The following technical lemma shows how to remove quantifiers from the update formulas of the query relation. For an arbitrary quantifier prefix \mbox{$\quant \in \{\exists, \forall\}^*$} let $\QFO$ be the class of queries expressible by formulas with quantifier prefix $\quant$. If $\quant$ is a substring of $\quant'$ and $\query$ is a query in $\QFO$ then trivially $\query$ is in $\QFO[\quant']$ as well.

\begin{lemma}\label{lemma:shifting}
  Let $\quant$ be an arbitrary quantifier prefix. For every $\DynQFO$-program there is an equivalent $\DynQFO$-program $\prog$ such that the update formulas for the designated query symbol of $\prog$ consist of a single atom.
\end{lemma}
\begin{proof}
  We follow the approach from \mbox{Example \ref{example:shifting}}. For ease of presentation we fix the input schema to be $\inpSchema = \{E\}$ where $E$ is a binary relation symbol; the proof can be easily adapted to arbitrary input schemas.

  Let $\prog$ be a \DynC-program over auxiliary schema $\schema$ with designated query symbol $Q$. We construct an equivalent $\DynC$ program $\prog'$ over schema $\schema'$ where $\schema'$ contains a designated query symbol $Q'$ and a $(k+2)$-ary relation symbol $R_\delta$ for every $k$-ary  $R \in \schema$  and every $\delta \in \{\ins, \del\}$.

  The idea is that $R_\delta$ shall reflect the content of $R$ in the next state, for each possible modification of the kind $\delta$. More precisely, let $G = (E, V)$ be a graph, $\alpha$ a sequence of modifications, $\beta = \delta(\vec e)$ a modification with $\delta \in \{\ins, \del\}$ and $\vec e \in V^2$. If $\state$ is the state obtained by $\prog$ after applying $\alpha \beta$ to $G$, i.e. $\state = \updateState{\prog}{\alpha \beta}{\init(G)}$, and $\state'$ is the state obtained by $\prog'$ after applying $\alpha$ to $G$, i.e. $\state' = \updateState{\prog'}{\alpha}{\init'(G)}$, then
    \begin{equation} \label{eq:shifting}
      \text{$\vec a \in R^{\state}$ if and only if $(\vec e, \vec a) \in R^{\state'}_\delta$.}
    \end{equation}
  Thus for every $\delta(\vec e)$ the relation $R_\delta(\vec e, \cdot)$ stores $R(\cdot)$ after application of $\delta(\vec e)$.  

  Then the update formula for $Q'$ after a modification $\delta$ can be written as follows:
    $$\uf{Q'}{\delta}{\vec u}{\vec x} \df R_\delta(\vec u, \vec x)$$

  It remains to explain how to update the relations $R_\delta$. We use formulas $\ufwa{E}{\ins}$ and $\ufwa{E}{\del}$ that express the impact of a modification to $E$, for example, \mbox{$\uf{E}{\ins}{a,b}{x,y} = E(x,y) \vee (a=x \wedge b = y)$}. By \mbox{$\ufsubstitute{\schema \rightarrow \schema_{\delta_0}}{R}{\delta_1}{\vec u_0}{\vec u_1, \vec x}$} we denote the formula obtained from $\uf{R}{\delta_1}{\vec u_1}{\vec x}$ by replacing every atom $S(\vec z)$ with $S \in \schema$ by $S_{\delta_0}(\vec u_0, \vec z)$.
  Then the update formula for  $R_{\delta_1}$ is 
    $$\uf{R_{\delta_1}}{\delta_0}{\vec u_0}{\vec u_1, \vec x} \df \ufsubstitute{\schema \rightarrow \schema_{\delta_0}}{R}{\delta_1}{\vec u_0}{\vec u_1, \vec x}.$$

  The initialization mapping of $\prog'$ is as follows. The query symbol $Q'$ is initialized like $Q$ in $\prog$. For every graph $G$ the relation symbol $R_\delta \in \schema'$ is initialized as
    $$ \bigcup_{\vec e \in V^2} \{\vec e\} \times \updateRelation{\prog}{\delta(\vec e)}{\init(G)}{R_\delta}$$
  where $\updateRelation{\prog}{\delta(\vec e)}{\init(G)}{R_\delta}$ denotes the relation $R_\delta$ in state $\updateState{\prog}{\delta(\vec e)}{\init(G)}$.
  
  The correctness of this construction is proved inductively over the length of modification sequences by showing that states of $\prog'$ \emph{simulate} states of $\prog$ as specified by (\ref{eq:shifting}).

  Therefore, let $G$ be a graph and $\alpha = \alpha_1 \ldots \alpha_i$ a modification sequence with $\alpha_i = \delta_i(\vec e_i)$. Further let $\state_j$ and $\state'_j$ be the states obtained by $\prog$ and $\prog'$, respectively, after application of $\alpha_1 \ldots \alpha_i$.
  
  If $\alpha$ is of length $0$ and $\beta$ is an arbitrary modification with $\beta = \delta(\vec e)$ then $\state \df \updateState{\prog}{\beta}{\state_0}$ and $\state' \df \state'_0$ satisfy (\ref{eq:shifting}) thanks to the definition of the initialization mapping \mbox{of $\prog'$}. 

  If $\alpha$ is of length $i \geq 1$ then, by induction hypothesis, the states \mbox{$\state \df \state_i = \updateState{\prog}{\alpha_i}{\state_{i-1}}$} and $\state' \df \state'_{i-1}$ satisfy (\ref{eq:shifting}) that is 
  \begin{equation}\label{eq:shiftinga}
    \text{$\vec a \in R^\state$ if and only if $(\vec e_i, \vec a) \in R^{\state'}_{\delta_i}$}
  \end{equation}
 for all relations $R$ and $R_{\delta_i}$.

  Now, let $\beta = \delta(\vec e)$ be an arbitrary modification. Further let $\calT \df \updateState{\prog}{\beta}{\state}$ and $\calT' \df \updateState{\prog'}{\alpha_i}{\state'}$. By definition, $\vec b \in R^\calT$ if and only if
     $$(R^\state, \{\vec u_1 \mapsto \vec e, \vec x \mapsto \vec b\}) \models \uf{R}{\delta}{\vec u_1}{\vec x}.$$
  Thanks to (\ref{eq:shiftinga}) and the definition of $\ufwa{R_{\delta_1}}{\delta_0}$ this is equivalent to 
     $$(R^{\state'}, \{\vec u_0 \mapsto \vec e_i, \vec u_1 \mapsto \vec e, \vec x \mapsto \vec b\}) \models \ufsubstitute{\schema \rightarrow \schema_{\delta_0}}{R}{\delta_i}{\vec u_0}{\vec u_1, \vec x}.$$
   By definition this is equivalent to $(\vec e, \vec b) \in R^{\calT'}$.
\end{proof}

Now we turn towards constructions for removing negations and unions from certain fragments. We start by exhibiting negation-free normal forms for \DynFO and for \DynProp.
\begin{lemma} \label{lemma:negationfree}
  \begin{enumerate}
    \item Every \DynFO-program has an equivalent negation-free \DynFO-program.
    \item Every \DynProp-program has an equivalent  \mbox{\DynPropUCQ-program}.
  \end{enumerate}
\end{lemma}
\begin{proof}
  This theorem is a generalization of Theorem 6.6 from \cite{Hesse03}. Given a dynamic program $\prog$, the simple idea is to maintain, for every auxiliary relation $R$ of $\prog$, an additional auxiliary relation $\Rh$ for the complement of $R$.

  We make this more precise. In the following we prove (a). As the construction does not introduce quantifiers, it can be used for (b) as well.

  Let $\calP = (P, \init, Q)$ be a \DynFO-program over schema $\schema$. We assume, without loss of generality, that $\calP$ is in negation normal form. Further we assume, for ease of presentation, that the input relations have update formulas as well\footnote{E.g. if the input database is a graph, then $\uf{E}{\ins}{a,b}{x,y} = E(x,y) \vee (a=x \wedge b = y)$ etc.}. 

  We construct a negation-free \DynFO-program equivalent to $\calP$ that uses the schema $\schema \cup \symneg{\schema} \cup \{\eqh\}$ where $\symneg{\schema}$ contains for every relation symbol $R \in \schema$ a fresh relation symbol $\symneg{R}$ of equal arity. Recall that $\schema$ includes the input schema and the auxiliary schema. The idea is to maintain in $\symneg{R}^\state$ the negation of $R^\state$, for all states $\state$. Further $\eqh$ shall always contain the complement of $=$. 

  In a first step we construct a \DynFO-program $\calP' = (P', \init', Q)$ in negation normal form over $\schema \cup \symneg{\schema} \cup \{\eqh\}$ that maintains $R^\state$ and $\symneg{R}^\state$ (but still uses negations). The update formulas for relation symbols $R \in \schema$ are as in $\calP$. For every $\symneg{R} \in \symneg{\schema}$ and every modification $\delta$, the update formula $\uf{\symneg{R}}{\delta}{\vec x}{\vec y}$ is the negation normal form\footnote{Observe that this fails for $\DynCQneg$ and $\DynUCQneg$.} of $\neg \uf{R}{\delta}{\vec x}{\vec y}$. The relation $\eqh$ never changes. The initialization mapping $\init'$ initializes $\symneg{R}$ with the complement of $\init(R)$.

  From $\calP'$ we construct a negation-free \DynFO-program $\calP'' = (P'', \init', Q)$. An update formula $\uf{R}{\delta}{\vec x}{\vec y}$ for $\calP''$ is obtained from the update formula $\uf{R}{\delta}{\vec x}{\vec y}$ for $\calP'$ by replacing all negative literals $\neg S$ by $\symneg{S}$. The initialization mapping of $\calP''$ is the same as for $\calP'$. 

  The equivalence of $\calP$ and $\calP''$ can be proved by an induction over the length of modification sequences.

\end{proof}

Two different techniques are used for removing unions. We start by giving a disjunction-free normal form for \DynProp.

\begin{lemma}\label{lemma:dynpropcq}
  Every \DynProp-program has an equivalent \DynAndNeg-program.  
\end{lemma}
\begin{proof}
    Let $\calP = (P, \init, Q)$ be a \DynProp-program over schema $\schema$. We assume, without loss of generality, that $\schema$ contains, for every relation symbol $R$, a relation symbol $\symneg{R}$ and that $\calP$ ensures that  $\symneg{R}^\state$ is the complement of $R^\state$ for every state $\state$. This can be achieved by using the same technique as in \mbox{Lemma \ref{lemma:negationfree}}. Further we assume that all update formulas of $\calP$ are in conjunctive normal form.

  The conjunctive \DynProp-program we are going to construct is over schema $\schema \cup \schema'$ where $\schema'$ contains a fresh relation symbol $R_{\neg C}$ for every clause $C$ occurring in some update formula of $\calP$.  The goal of the construction is to ensure that $R_{\neg C}^\state(\vec z)$ holds if and only if $\neg C(\vec z)$ is true in state $\state$. Then an update formula $\phi=C_1(\vec x_1) \wedge \ldots \wedge C_k(\vec x_k)$ with clauses $C_1(\vec x_1), \ldots, C_k(\vec x_k)$ can be replaced by the conjunctive formula $\neg R_{\neg C_1}(\vec x_1) \wedge \ldots \wedge \neg R_{\neg C_k}(\vec x_k)$.

  In a first step we construct a \DynProp-program $\calP' = (P', \init', Q)$ in conjunctive normal form that maintains the relations $R_{\neg C}^\state$. To this end, let $C$ be a clause with $k$ variables and let $\vec z$ be the $k$-tuple that contains those variables in the order in which they occur in $C$. Assume that \mbox{$C = L_1(\vec z_1) \vee \ldots \vee L_l(\vec z_l)$} where $\vec z_i \subseteq \vec z$ and each $L_i$ is an atom or a negated atom. Thus \mbox{$\neg C \equiv \neg L_1(\vec z_1) \wedge \ldots \wedge \neg L_l(\vec z_l)$}. The relation symbol $R_{\neg C}$ is of arity $k$. For a modification $\delta$ the update formula for $R_{\neg C}$ is $$\uf{R_{\neg C}}{\delta}{\vec x}{\vec z} = \uf{X_1}{\delta}{\vec x}{\vec z_1} \wedge \ldots \wedge \uf{X_l}{\delta}{\vec x}{\vec z_l}$$ where $X_i$ is the relation symbol $R$ if $L_i = \neg R$ and $X_i$ is $\symneg{R}$ if $L_i = R$. Observe that $\uf{R_{\neg C}}{\delta}{\vec x}{\vec z}$ is in conjunctive normal form, because each $\uf{X_i}{\delta}{
\vec x}{\vec z_i}$ is in conjunctive normal form; further $\uf{R_{\neg C}}{\delta}{\vec x}{\vec z}$ does not use new clauses. The initialization mapping $\init'$ extends the initialization mapping $\init$ to the schema $\tau'$ in a natural way. For a clause $C$ and input database $\inp$, a tuple $\vec a$ is in $\init'(R_{\neg C})$ if and only if $C$ evaluates to false in $\init(\inp)$ for $\vec a$.

  The second step is to construct from $\calP'$ the desired conjunctive \DynProp-program $\calP''$: every clause $C$ in every update formula of $\calP'$ is replaced by $\neg R_{\neg C}$. This construction yields a conjunctive program $\calP''$. The initialization mapping of $\calP''$ is the same as for $\calP'$. 

  We sketch the proof that $\calP''$ is equivalent to $\calP$. The dynamic program $\calP'$ updates relations from $\schema$ exactly as program $\calP$. By an induction over the length of modification sequences, one can prove that $R_{\neg C}^\state(\vec a)$ holds if and only if $\neg C(\vec a)$ is true in state $\state$ for all tuples $\vec a$. Thus corresponding update formulas of $\calP$ and $\calP''$ always yield the same result.

\end{proof}

Now we turn to disjunction-free normalforms for $\DynUCQ$, $\DynUCQneg$ and negation-free $\DynFO$. Observe that the idea of the proof of Lemma \ref{lemma:dynpropcq} cannot be applied directly since those classes are not closed under negations. Instead the idea is to simulate disjunctions by existential quantifiers. 
\begin{lemma}\label{lemma:disjunctionfree}
  \begin{enumerate}[(a)]
    \item For every $\DynUCQneg$-program there is an equivalent $\DynCQneg$-program.    \item For every $\DynUCQ$-program there is an equivalent $\DynCQ$-program.
    \item For every $\DynFO$-program there is an equivalent $\DynFOand$-program.
  \end{enumerate}
\end{lemma}
\begin{proof}
  We first prove the statements for domains with at least two elements and show how to drop this restriction afterwards. The construction uses the idea from Example \ref{example:disjunctionfree}.
We give it for (a) but, as it does not introduce any negation operators it works for (b) as well. For (c) it is sufficient to start from a negation-free $\DynFO$-program by Lemma \ref{lemma:negationfree}; and for those the same construction as for (a) can be used\footnote{More precisely, replace the quantifier prefix $\exists \vec y$ used throughout the construction of (a) by a general quantifier-prefix $\exists \vec y_1 \forall \vec y_2 \ldots$.}.

  Let $\prog = (P, \init, Q)$ be a $\DynUCQneg$-program over schema $\schema$. Without loss of generality, we assume that the quantifier-free parts of all update formulas of $\prog$ are in disjunctive normal form. We convert $\prog$ into an equivalent $\DynCQneg$-program $\prog'$ whose update formulas are in prenex normalform with quantifier-free parts of the form $\bigwedge_i L_i(\vec x_i) \land T(\vec y)$, where $L_i$ are arbitrary literals over a modified schema $\schemah$ and the symbols $T$ are fresh auxiliary relation symbols. 
  The program $\prog' = (P', \init', Q)$ is over schema $\schema' = \schema \cup \schemah \cup \schema_T$, where $\schemah$ contains a $(k+1)$-ary relation symbols $\Rh_1$ and $\Rh_2$ for every $k$-ary relation symbol $R \in \schema$; and $\schema_T$ contains a relation symbol $T_{S, \delta}$ for every relation symbol $S \in \schema \cup \schemah$ and every modification $\delta$. 

  The intention for relation symbols from $\schema'$ is as follows. The relation symbols $R \in \schema$ shall always be interpreted as in $\prog$. The intention of $\Rh_1 \in \schemah$ is, on one hand, to contain a ``copy of $R$'' (in tuples with first component $c$, for some fixed element $c$) and on the other hand, to guarantee non-emptiness. The latter  is strongly ensured by enforcing all tuples that do \emph{not} have $c$ as first component to be in $\Rh$ and by $|\domain| \geq 2$. Similarly $\Rh_2 \in \schemah$ contains a ``copy of $R$'' but is not the universal relation, that is, it does not contain all tuples. More precisely: 
    \begin{equation}\label{eq:disjunctionfree}
      \Rh_1^\state \df \{(c, \vec a) \mid \vec a \in R^\state\} \cup  \{(d, \vec a) \mid \text{$d \neq c$ and $\vec a \in \domain^k$}\}
    \end{equation}\vspace{-3mm}
    \begin{equation}\label{eq:disjunctionfreea}
      \Rh_2^\state \df \{(c, \vec a) \mid \vec a \in R^\state\}
    \end{equation}
  The relations $T_{S, \delta}$ will be used as in Example \ref{example:disjunctionfree}.

  Now we construct the update formulas for program $\prog'$. Let $R \in \schema$ and $\delta$ be a modification. Further let 
  $$\uf{R}{\delta}{\vec u}{\vec x} = \exists \vec y (C_1(\vec u, \vec x, \vec y) \lor \ldots \lor C_k(\vec u, \vec x, \vec y))$$
  be the update formula of $R$ with respect to $\delta$  in $\prog$, where every $C_i$ is a conjunction of literals. For 
    $$C_i(\vec u, \vec x, \vec y) = L_1(\vec v_1) \land \ldots \land L_m(\vec v_l)$$ 
  we define 
    $$\Ch_i(v, \vec u, \vec x, \vec y) \df \Lh_1(v, \vec v_1) \land \ldots \land \Lh_m(v, \vec v_l)$$ 
  where $\Lh_j = \Rh_1$ if $L_j = R$ and $\Lh_j= \neg \Rh_2$ if $L_j = \neg R$. 

  The  update formula $\ufb{\Rh_i}{\delta}{\vec u}{x', \vec x}$ for $\Rh_i \in \schemah$ in $\prog'$ is  
  \begin{align*}
    \ufb{\Rh_i}{\delta}{\vec u}{x', \vec x} \df &  \exists \vec y\;\exists z'_1 \exists \vec z_1 \ldots  \exists z'_k \exists \vec z_k\\
        & \quad \Big( \hat{C}_1(z'_1, \vec z_1) \land \ldots \land \hat{C}_k(z'_k, \vec z_k) \\ 
        & \quad \quad \quad \land T_{\Rh, \delta}(\vec y, z'_1, \vec z_1, \ldots, z'_k, \vec z_k, \vec u, x', \vec x) \Big).
  \end{align*}
  The  update formula $\ufb{R}{\delta}{\vec u}{\vec x}$ for $R \in \schema$ in $\prog'$ is
  \begin{align*}
    \ufb{R}{\delta}{\vec u}{\vec x} \df & \exists \vec y\;\exists z'_1 \exists \vec z_1 \ldots \exists z'_k \exists \vec z_k \\
    & \quad \Big(\hat{C}_1(z'_1, \vec z_1) \land \ldots \land \hat{C}_k(z'_k, \vec z_k) \\ 
    & \quad \quad \quad \land T_{R, \delta}(\vec y, z'_1, \vec z_1, \ldots, z'_k \vec z_k, \vec u, \vec x) \Big).
  \end{align*}

To ensure equivalence of this program with the original program, the relations $T_{S, \delta}$ are defined as follows.

\begin{itemize}
\item $T_{R, \delta}$ contains all tuples\footnote{For simplicity, we reuse variable names as element names.} $(\vec y, z'_1, \vec z_1, \ldots, z'_k, \vec z_k, \vec u, \vec x)$, for which, for some $i$, $z'_i=c$ and $\vec{z_i}= (\vec u, \vec x, \vec y)$.
\item $T_{\Rh_i, \delta}$ contains all tuples $(\vec y, z'_1, \vec z_1, \ldots, z'_k, \vec z_k, \vec u, x', \vec x) )$, for which
  \begin{itemize}
  \item $x'\not=c$, or
  \item for some $j$, $z'_j=c$ and $\vec{z_j}= (\vec u, \vec x, \vec y)$.
  \end{itemize}
\end{itemize}
  These are initialized as intended by simple quantifier-free formulas (but with disjunction). Their interpretation is never changed, that is, for every $T_{S, \delta}$, both update formulas reproduce the current value of  $T_{S, \delta}$.

  The initialization for relation symbols from $\schema$ and $\schemah$ is straightforward. Auxiliary relation symbols $R \in \schema$ are initialized as in $\prog$; and auxiliary relation symbols $\Rh_1,\Rh_2 \in \schemah$ are initialized by $\init'$ analogously to $\init$ but respecting \mbox{Equations (\ref{eq:disjunctionfree}) and (\ref{eq:disjunctionfreea})}.

  This concludes the proof of (a), (b) and (c) for domains with at least two elements.   
  The restriction on the size of the domains can be dropped as follows. In all three cases the idea is to make a case distinction on the size of the domain in the update formulas of the designated query symbol. 

  To this end, we first construct a $\DynPropCQ$-program $\prog'' = ({P''}, {\init''}, {Q''})$ over schema ${\schema''}$ with $\schema' \cap \schema'' = \emptyset$ which is equivalent to $\prog$ over databases with domains of size one. Then we construct a program $\prog'''$ equivalent to $\prog$ by combining the programs $\prog'$ and $\prog''$. 

  For the construction of $\prog''$ we observe that every relation of a database over a single element domain $\domain = \{a\}$ contains either exactly one tuple, namely $(a,\ldots, a)$, or no tuple at all. Thus every such relation $R$ corresponds to a $0$-ary relation $R_0$ where $R_0$ is true if and only if $(a, \ldots, a) \in R$. Hence, by \mbox{Lemma \ref{lemma:nullary}} (see below), there is a $\DynPropCQ$-program equivalent to $\prog$ for databases with domains of size one.
 
  To combine $\prog'$ and $\prog''$ we use two different approaches, one for (a) and one for (b) and (c). 

  First we consider (a). To this end, we can assume, by Lemma \ref{lemma:shifting}, that the update formulas for the query relations $Q'$ and $Q''$ of  $\prog'$ and $\prog''$, respectively, consist of single atoms. We construct an intermediate program $\tilde{\prog} = (\tilde{P}, \widetilde{\init}, \tilde{Q})$ over schema $\tilde{\schema} = \{\tilde{Q}, U\} \cup \schema' \cup \schema''$ where $U$ is a fresh $0$-ary relation symbol. The intention is that interpretations of symbols in $\schema'$ and $\schema''$ are as in $\prog'$ and $\prog''$, respectively, and that $U$ is interpreted by true if and  only if the domain is of size one. The initializations are accordingly.

  Thus all update formulas of $\tilde{\prog}$ for relation symbols from $\schema'$ and $\schema''$ are as in $\prog'$ and $\prog''$ (and thus disjunction-free). The update formula for $U$ is $\ufwa{U}{\delta} \df U$ and
    \begin{align*}
      \ufwa{\tilde{Q}}{\delta} & \df (\ufwa{Q'}{\delta} \land \neg U) \vee (\ufwa{Q''}{\delta} \land U) \\
      & \equiv (\ufwa{Q'}{\delta} \lor \ufwa{Q''}{\delta}) \wedge (\neg U \vee \ufwa{Q''}{\delta}) \wedge (\ufwa{Q'}{\delta} \lor U).     
    \end{align*}

  The program $\prog'''$ is obtained from $\tilde{\prog}$ by removing disjunctions from $\ufwa{\tilde{Q}}{\delta}$ using the method\footnote{This method cannot be used for $\DynCQ$ and $\DynFOand$.} from the proof of Lemma \ref{lemma:dynpropcq}. For example, the first clause is replaced by $\neg R_{\neg (Q' \vee Q'')}$ where $R_{\neg (Q' \vee Q'')}$ is a fresh auxiliary relation symbol intended to be always interpreted by the result of the query $\neg (\ufwa{Q'}{\delta} \lor \ufwa{Q''}{\delta})$. The update formula for $R_{\neg (Q' \vee Q'')}$ after a modification $\delta$ is $\neg \ufwa{Q'}{\delta} \wedge \neg \ufwa{Q''}{\delta}$; it is disjunction-free since, by our assumption, $\ufwa{Q'}{\delta}$ and $\ufwa{Q''}{\delta}$ both consist of a single atom. This concludes the proof \mbox{of (a)}.
  
  The program $\prog'''$ for (b) and (c) is over schema $\schema''' = \{Q'''\} \cup \schema' \cup \schema''$. Again all update formulas of $\prog'''$ for relation symbols from $\schema'$ and $\schema''$ are as in $\prog'$ and $\prog''$ and
    $$\ufwa{Q'''}{\delta} = \ufwa{Q'}{\delta} \land \ufwa{Q''}{\delta}.$$ 
  The case distinction is delegated to the initialization mapping. Recall that the size of the domain is fixed when the auxiliary relations are initialized. The initialization mapping $\init'''$ is as follows.
  If $|\domain| = 1$ then 
        $$\init'''(R) = \begin{cases}
          {\init''}({Q''}) & \text{for $R = Q'''$, }\\
         D^k & \text{for $R \in \schema'$,}\\
          {\init''}({R''}) & \text{for $R \in \schema''$} 
                    \end{cases}$$
  If $|\domain| \geq 2$ then 
      $$\init'''(R) = \begin{cases}
          \init'(Q') & \text{for $R = Q'''$,} \\
          \init'(R') & \text{for $R \in \schema'$,} \\
                      D^k & \text{for $R \in {\schema''}$}
      \end{cases}$$

  Thus $\init'''$ selects either $\ufwa{Q'}{\delta}$  or $\ufwa{{Q''}}{\delta}$, depending on the size of the domain. If $|\domain| = 1$ then $\ufwa{Q'}{\delta}$ always evaluates to true whereas $\ufwa{{Q''}}{\delta}$ yields the same value as in ${\prog''}$, and vice versa for $|\domain| \geq 2$. As update formulas do not use negation, all relations in the program, that is initialized to ``true'' ($\prog'$ or $\prog''$) remain ``full'' throughout.\footnote{This cannot be guaranteed for $\DynUCQneg$.} This concludes the proof of (b).
\end{proof}

  It remains to prove that all queries over $0$-ary relations can be maintained in $\DynPropCQ$. $0$-ary relations can either be true (containing the empty tuple) or false (not containing the empty tuple and thus being empty), thus $0$-ary atoms are basically propositional variables. Queries on $0$-ary databases are therefore basically families of Boolean functions, one for each domain size. Such queries are not very interesting from the perspective of databases, but we need to show the following lemma as we used it in the previous proof. 
  
  As quantification in queries on $0$-ary databases is useless, every $\FO$ query can be expressed by a quantifier-free formula and therefore  can be maintained in $\DynProp$. The following lemma shows that this can be sharpened.

  \begin{lemma} \label{lemma:nullary}
    Every query on a $0$-ary database can be maintained by a $\DynAND$-program.
  \end{lemma}
  \begin{proof}
      Let $\inpSchema$ be an input schema with $0$-ary relation symbols $A_1, \ldots, A_k$. Further let $\calQ_1, \ldots, \calQ_{m}$ be an enumeration of all $m = 2^{2^k}$ many queries on $\inpSchema$. We actually show that all of them can be maintained by one $\DynAND$-program $\prog$ with auxiliary schema $\auxSchema = \{R_1, \ldots, R_{m}\}$ maintaining $\calQ_i$ in $R_i$, for every $i \in\{1,\ldots,m\}$.
      To this end, let $\varphi_1, \ldots, \varphi_{m}$ be propositional formulas over $\inpSchema$  such that $\varphi_i$ expresses $\calQ_i$ and each $\varphi_i$ is in conjunctive normal form. Without loss of generality, no clause contains $A_l$ and $\neg A_l$ for any $A_l \in \inpSchema$ and any $\varphi_i$.  As $\auxSchema$ contains a relation symbol, for every propositional formula over $A_1, \ldots, A_k$, it contains, in particular,  an auxiliary relation symbol $R_C$, for every disjunctive clause over $A_1, \ldots, A_k$. 

      The update formulas for $R_j$ after changing input relation $A_l$ can be constructed as follows. Let $\calC$ be the set of clauses of $\varphi_j$, i.e. $\varphi_j =  \bigwedge_{C \in \calC} C$. We denote by $\calC^+_{A_l}$, $\calC^-_{A_l}$ and $\calC_{A_l}$ the subsets of $\calC$ whose clauses contain $A_l$, $\neg A_l$ and neither $A_l$ nor $\neg A_l$, respectively. 

      If $A_l$ becomes true by a modification then $\varphi_j$ evaluates to true if all clauses in $\calC_{A_l}$ and all clauses $C \setminus \{\neg A_l\}$ with $C \in \calC^-_{A_l}$ evaluated to true before the modification (clauses in $\calC^+_{A_l}$ will evaluate to true after enabling $A_l$). 

      If $A_l$ becomes false by a modification then $\varphi_j$ evaluates to true if all clauses in $\calC_{A_l}$ and all clauses $C \setminus \{A_l\}$ with $C \in \calC^+_{A_l}$ evaluated to true before the modification (clauses in $\calC^-_{A_l}$ will evaluate to true after disabling $A_l$). 

      Therefore the update formulas for $R_j$ after updating $A_l$ can be defined as follows: 
      $$\ufwa{R_j}{\ins_{A_l}} \df \bigwedge_{C \in \calC_{A_l}} R_{C} \land \bigwedge_{C \in \calC_{A_l}^-} R_{C \setminus \{\neg A_l\}}$$
      $$\ufwa{R_j}{\del_{A_l}} \df \bigwedge_{C \in \calC_{A_l}} R_{C} \land \bigwedge_{C \in \calC_{A_l}^+} R_{C \setminus \{A_l\}}$$

      The initialization is straightforward. The correctness of this construction can be proved by induction over the length of modification sequences.
  \end{proof}

Finally we prove that $\DynEFO = \DynAFO$, and therefore that unions of conjunctive queries with negation coincide with $\DynAFO$ in the dynamic setting. The proof uses the replacement technique  to maintain the complements of the auxiliary relations used in the  $\DynEFO$-program via $\DynAFO$-formulas. A small complication arises from the fact, that the query relation (and not its complement) has to be maintained. This is solved by ensuring that the update formulas of the query relation are atomic.

A slightly more general result can be shown.
\begin{lemma}\label{lemma:quantifierswitch}
  Let $\quant$ be an arbitrary quantifier prefix. A query can be maintained in $\DynQFO$ if and only if it  can be maintained in $\DyncQFO$.
\end{lemma}
\begin{proof}
  Let $\prog = (P, \init, Q)$ be an arbitrary dynamic $\DynQFO$-program over \mbox{schema $\schema$}. By Lemma \ref{lemma:shifting} we can assume, without loss of generality, that the update formulas of $Q$ are atomic. We construct a dynamic $\DyncQFO$-program $\prog'$ over schema $\schemah \cup \{Q'\}$ where $\schemah$ contains a $k$-ary relation symbol $\Rh$ for every $k$-ary $R \in \schema$. The intention is that $\Rh$ is always equal to the complement of $R$.

  We denote by $\substitutewa{\phi}{\schema \rightarrow \schemah}$ the formula obtained from $\phi$ by replacing every atom $S(\vec z)$ in  $\phi$ by $\neg \Sh(\vec z)$.  Then the update formulas of $\prog'$ are obtained as $\ufwa{\Rh}{\delta} \df \neg \ufsubstitutewa{\schema \rightarrow \schemah}{R}{\delta}$ for every $\Rh \in \schemah$. Observe that this formula can be easily transformed into an $\cQFO$-formula. Further $\ufwa{Q'}{\delta} = \neg \ufwa{\Qh}{\delta}$ which is a $\cQFO$-formula since $\ufwa{\Qh}{\delta}$ is quantifier-free.  The initialization mapping of $\prog'$ is straightforward.
\end{proof}

Now Theorems \ref{theorem:dynefoequivalences}, \ref{theorem:dynucqequivalences} and \ref{theorem:dynpropequivalences}  follow immediately from Lemmata \ref{lemma:negationfree}, \ref{lemma:dynpropcq}, \ref{lemma:disjunctionfree} and \ref{lemma:quantifierswitch}.
 
    \subsection{Separation results}\label{section:separations}
      The dynamic complexity classes $\DynAND$ and $\DynProp$ have been separated already in \cite[Lemma 7.4]{ZeumeS15reach} using the non-empty-set query. Here we extend this result to the following theorem.

\begin{theorem}\label{theorem:separations}
  \begin{enumerate}
    \item The class $\DynAND$ is a strict subclass of $\DynProp$.
    \item The class $\DynProp$ is a strict subclass of $\DynCQ$.
  \end{enumerate}
\end{theorem}

For the sake of completeness we repeat the proof idea for the separation of $\DynAND$ and $\DynProp$ as used in \cite{ZeumeS15reach}.

\begin{proofsketchof}{Theorem \ref{theorem:separations}(a)}
  Towards a contradiction assume that there is a $\DynAnd$-program $\prog = (P, \init, Q)$ 
that maintains the query $\query$ defined by $\exists x U(x)$. By possibly introducing additional auxiliary relations we can guarantee that no variable occurs more than once in any atom of an update \mbox{formula of $\prog$}. 

  The following observation is the key to the proof. Let us assume that there is a unary relation atom $R(u)$ in the formula $\uf{Q}{\del}{u}{}$ for the 0-ary query relation $Q$ and that $\state$ is a state in which the relation $U$ contains at least two elements $a\not=b$. Then, necessarily, $R^\state$ contains both, $a$ and $b$, as otherwise deletion of $a$ or $b$ would make $\ufwa{Q}{\del}$ false without $U$ becoming empty. This observation can be generalized:
if a relation $R$ has \emph{deletion depth} $k$, that is, the distance of $R$ from $Q$ in the deletion dependency graph is $k$, and $U$ contains at least $k+1$ elements, then $R$ must contain all \emph{diverse} tuples over $U$, that is, tuples that consist of pairwise distinct elements from $U$. 

  Now the idea is as follows. Consider a state $\state$ in which the set $U$ contains $m+1$ elements, where $m$ is the maximum (finite) deletion depth of any relation symbol in $\prog$. By the observation above, all relations whose symbols are reachable from $Q$ in the deletion dependency graph of $\prog$ contain all diverse tuples over $U^\state$. Thus, all relation atoms over diverse tuples from $U^\state$ evaluate to true. It is easy to show by induction on the length of modification sequences  that this property holds  (with respect to $U^{\state'}$) for all states $\state'$ that can be obtained from $\state$ by deleting elements from $U^\state$. In particular, it holds for any such state in which $U^{\state'}$ contains only one element $a$. But then, $\uf{Q}{\del}{a}{}$ evaluates to true in $\state'$ and thus $Q$ remains true after deletion of $a$, the desired contradiction to the assumed correctness \mbox{of $\prog$}.
\end{proofsketchof}

To prove the second statement of Theorem \ref{theorem:separations} we show that the class $\DynQF$ (to be defined in a moment) is a subclass of  $\DynCQ$. This proves Theorem \ref{theorem:separations}(b) as it is known from \cite{GeladeMS12} that $\DynProp$ is strictly contained in \DynQF.

We give an informal account of \DynQF before the actual definition and the statement of the theorem. The class $\DynQF$ is an extension of $\DynProp$, and was introduced by Hesse in \cite{Hesse03}. Update formulas in $\DynQF$ are quantifier-free (as in $\DynProp$) but auxiliary functions may be used. Thus, while  \DynProp formulas can only access the inserted or deleted tuple $\vec a$ and the currently updated tuple $\vec b$ of an auxiliary relation, \DynQF update formulas can access further elements of the universe obtained by applying auxiliary functions to elements of $\vec a$ and $\vec b$. Further, upon modification of the input database, auxiliary functions may be updated via update terms that may use functions and if-then-else-constructs.

The following formalization is adapted from \cite{GeladeMS12} and \cite{ZeumeS15reach}. We extend our definition of schemata to allow also function symbols. Until the end of this section, a \emph{schema (or signature)} $\schema$ consists of a set $\relSchema$ of relation symbols and a set $\funSchema$ of function symbols together with an arity function $\arity: \relSchema\cup\funSchema \mapsto \N$. A schema is \emph{relational} \mbox{if $\funSchema=\emptyset$}.  A \emph{database} $\db$ of schema $\schema$ with domain $\domain$ is a mapping that assigns to every relation symbol $R \in \relSchema$ a relation of arity $\arity(R)$ over $\domain$ and to every $k$-ary function symbol $f \in \funSchema$ a $k$-ary function.

Next, we extend our definition of update programs to auxiliary schemas with functions. For updating auxiliary functions case distinctions are allowed in addition to the composition of function terms.

\begin{definition}(Update term)
  \textit{Update terms} are inductively defined as follows:
  \begin{itemize}
  \item[(1)] Every variable is an update term.
  \item[(2)] If $f$ is a $k$-ary function symbol and $t_1, \dots, t_k$
    are update terms, then $f(t_1, \ldots, t_k)$ is an update term.
  \item[(3)] If $\phi$ is a quantifier-free update formula
    (possibly using update terms) and $t_1$ and $t_2$ are update
    terms, then $\ite{\phi}{t_1}{t_2}$ is an update term.
  \end{itemize} 
The semantics of update terms associates with every
  update term $t$ and interpretation $I=(\state,\beta)$, where $\state$
  is a state and $\beta$ a variable assignment, a value $\sem{t}{I}$ from
  $S$.  
  The semantics of (1) and (2) is straightforward. If $\state\models\phi$ holds, then $\sem{\ite{\phi}{t_1}{t_2}}{I}$
  is  $\sem{t_1}{I}$, otherwise $\sem{t_2}{I}$.
\end{definition}

The extension of the notion of update programs for auxiliary schemas with function symbols is now straightforward. An update program still has an update \mbox{formula $\ufwa{R}{\delta}$} (possibly using update terms) for every relation symbol $R \in \auxSchema$ and every
  abstract modification $\delta$. Furthermore, it has, for every
  abstract modification $\delta$ and  every function symbol $f \in \auxSchema$, an update term $\ut{f}{\delta}{\vec x}{\vec y}$. For a concrete modification $\delta(\vec a)$ it redefines $f$ for each tuple $\vec b$ by evaluating  $t^f_\delta(\vec a;\vec b)$ in the current state.

  \begin{definition}(\DynQF)
    \DynQF is the class of queries maintainable by quantifier-free update programs with (possibly) auxiliary functions. 
  \end{definition}

  We remark that our definition of $\DynQF$ is slightly stronger than the usual definition. Here we allow for using update terms in update formulas for relations whereas in  \cite{GeladeMS12} only terms are allowed. This strengthens the result we are aiming at. 

 Before continuing, we give a toy example designed to give an impression of the expressive power of $\DynQF$. For more examples we refer to  \cite{Hesse03} and \cite{GeladeMS12}.

\begin{example}
  Consider the unary graph query $\query(x)$ that returns all nodes $a$ of a given graph $G$ with maximal outdegree .

  We construct a unary $\DynQF$-program $\prog$ that maintains $\query$ in a unary relation denoted by the designated symbol $Q$.  The program uses two unary functions $\Succ$ and $\Pred$ that shall encode a successor and its corresponding predecessor relation on the domain. For simplicity, but without loss of generality, we therefore assume that the domain is of the form $\domain = \{0, \ldots, n-1\}$. For every state $\state$, the function $\Succ^\state$ is then the standard successor function \mbox{on $\domain$} (with $\Succ^\state(n-1) = n-1$), and $\Pred^\state$ is the standard predecessor function (with $\Pred^\state(0) = 0$). Both functions are initialized accordingly. In the following we refer to \emph{numbers} and mean the position of elements \mbox{in $\Succ$}. The program has a constants that represents the numbers $0$ and $1$.
  
  The program $\prog$ maintains two unary functions $\numEdges$ and $\numNodes$. The function $\numEdges$ counts, for every node $a$, the number of outgoing edges of $a$; more precisely $\numEdges(a) = b$ if and only if $b$ is the number of outgoing edges \mbox{of $a$}. The function $\numNodes$ counts, for every number $a$, the number of nodes with $a$ outgoing edges; more precisely $\numNodes(a) = b$ if and only if $b$ is the number of nodes with $a$ outgoing edges. A constant $\Max$ shall always point to the number $i$ such that $i$ is the maximal number of outgoing edges from some node in the current graph.
  
  When inserting an outgoing edge $(u,v)$ for a node $u$ that already has $a$ outgoing edges, the counter $\numEdges$ of $u$ is incremented from $a$ to $a+1$ and all other edge-counters remain unchanged. The counter $\numNodes$ of $a$ is decremented, the counter of $a+1$ is incremented, and all other node-counters remain unchanged. The number $\Max$ increases if, before the insertion, $u$ was a node with maximal number of outgoing edges. This yields the following update terms:
  \begin{align*}
    \ut{\numEdges}{\ins\;E}{u,v}{x} \df & \itewa\Big(\neg E(u,v) \wedge x=u,\Succ(\numEdges(x)),\numEdges(x)\Big) \\
    \ut{\numNodes}{\ins\;E}{u,v}{x} \df & \itewa\Big(\neg E(u,v) \wedge x = \numEdges(u), \Pred(\numNodes(x)), \\
      & \quad \itewa\big(\neg E(u,v) \wedge x = \Succ(\numEdges(u)), \Succ(\numNodes(x)), \\
      & \quad \quad \numNodes(x)  \big)\Big) \\
    \utw{\Max}{\ins\;E}{u,v} \df &  \itewa\Big(\Max = \numEdges(u) \wedge \neg E(u,v), \Succ(u),\Max \Big) 
  \end{align*}

  The update formula for the designated query symbol $Q$ is as follows:
   \begin{align*}
    \uf{Q}{\ins\;E}{u,v}{x} \df  \ut{\numEdges}{\ins\;E}{u,v}{x} = \utw{\Max}{\ins\;E}{u,v}
   \end{align*}
   
  The update terms for deletions are very similar:
  \begin{align*}
    \ut{\numEdges}{\del\;E}{u,v}{x} \df & \itewa\Big(E(u,v) \wedge x=u,\Pred(\numEdges(x)),\numEdges(x)\Big) \\
    \ut{\numNodes}{\del\;E}{u,v}{x} \df & \itewa\Big(E(u,v) \wedge x = \numEdges(u), \Pred(\numNodes(x)), \\
      &  \quad  \itewa\big(E(u,v) \wedge x = \Pred(\numEdges(u)), \Succ(\numNodes(x)), \\
      &  \quad \quad\numNodes(x)  \big)\Big) \\
    \utw{\Max}{\del\;E}{u,v} \df &  \itewa\Big(\Max = \numEdges(u) \wedge E(u,v) \wedge  \numNodes(\Max) = 1, \\
    &  \quad \Pred(\Max),\Max \Big) 
  \end{align*}
    
  The update formula for the designated query symbol $Q$ under deletion is as follows:
   \begin{align*}
    \uf{Q}{\del\;E}{u,v}{x} \df  \ut{\numEdges}{\del\;E}{u,v}{x} = \utw{\Max}{\del\;E}{u,v}    
   \end{align*}
    \qed
\end{example}

We aim at the following theorem.  
  
\begin{theorem} \label{theorem:DynCQcontainsDynQF}
  \DynQF is contained in \DynCQ.
\end{theorem}

 Thanks to Theorem \ref{theorem:dynucqequivalences} it suffices to show that \DynQF is contained in \DynUCQ. The idea of the proof of Theorem \ref{theorem:DynCQcontainsDynQF} is to simulate auxiliary functions by auxiliary relations with the help of existential quantifiers in a relatively straightforward way. However, some care is necessary in order to remove $\ite{}{}{}$-conditions and negations. We highlight  the idea of the  proof by the following example.
  
  \begin{example}
    Consider a $\DynQF$-program $\prog$ that contains the following update term $\utwa{f}{\ins E}$ for a unary function $f$ and update formula $\ufwa{R}{\ins E}$ for a unary \mbox{relation $R$}:
      \begin{align*}
        \ut{f}{\ins E}{u,v}{x} &\df f(\ite{R(x)}{f(x)}{u}) \\
        \uf{R}{\ins E}{u,v}{x} &\df \neg R(x) \wedge S(f(x), \ite{\neg R(\ite{S(u,v)}{u}{x})}{f(x)}{u})
      \end{align*}
    As a first step towards the construction of an equivalent $\DynCQ$-program, we remove negations by maintaining for every relation $T$ its complement in an auxiliary relation $\Th$, for example:
      \begin{align*}
        \uf{R}{\ins E}{u,v}{x} &\df \Rh(x) \wedge S(f(x), \ite{\Rh(\ite{S(u,v)}{u}{x})}{f(x)}{u}) \\
        \uf{\Rh}{\ins E}{u,v}{x} &\df R(x) \vee \Sh(f(x), \ite{\Rh(\ite{S(u,v)}{u}{x})}{f(x)}{u})
      \end{align*}
      
    The crucial step in the construction of an equivalent $\DynCQ$-program is to simulate the function $f$ by a binary relation $R_f$ with the intention that $R_f$ stores all tuples of the form $(a, f(a))$. Then appearances of $f$ as well as of $\ite{}{}{}$ can be removed. The complement relations obtained in the preprocessing step above are also needed in order to remove $\ite{}{}{}$-conditions.
    
    The update formula $\uf{R_f}{\ins E}{u,v}{x, y}$ for $R_f$ is obtained by 'outsourcing' the computation of the $\ite{}{}{}$-value:
        $$\uf{R_f}{\ins E}{u,v}{x, y} \df \exists z \Big(R_f(z, y) \wedge \big(\underbrace{(R(x) \wedge R_f(x, z)) \vee (\Rh(x) \wedge z = u)}_{z = \ite{R(x)}{f(x)}{u}}\big) \Big)$$
    An update formula for $R$ that uses $R_f$ instead of $f$ is obtained similarly:
      \begin{align*}
          \uf{R}{\ins E}{u,v}{x} &\df \Rh(x) \wedge \exists z_1 \exists z_2 \exists z_3 \Big(S(z_1, z_2) \wedge R_f(x, z_1) & \\
          & \wedge \big((\Rh(z_3) \wedge R_f(x, z_2)) \vee (R(z_3) \wedge z_2 = u)\big) &  \rbrace \;\;  \ite{\Rh(\ldots)}{\cdot}{\cdot}\\
          & \wedge \big((S(u,v) \wedge z_3 = u) \vee (\Sh(u,v) \wedge z_3 = x)\big)\Big) &  \rbrace \;\; \ite{S(\ldots)}{\cdot}{\cdot}
      \end{align*}
    Observe that only relation symbols from the original $\DynQF$-program are needed in negated form.
  \end{example}

  \begin{proofof}{Theorem \ref{theorem:DynCQcontainsDynQF}}
      Let $\prog = (P, \init, Q)$ be a \DynQF-program over schema $\schema = \relSchema \cup \funSchema$. We assume, without loss of generality, that $\calP$ is in negation normal form. Further we assume, as in Lemma \ref{lemma:negationfree}, that the input relations have update formulas as well.
      
      We prove that there is a $\DynUCQ$-program $\prog''$  equivalent to $\prog$. Then, by Theorem \ref{theorem:dynucqequivalences}, there is an equivalent $\DynCQ$-program. 
      
      As a  preparation step we construct, from $\prog$, a $\DynQF$-program $\prog'$ over schema $\schema' \df \relSchema \cup \relSchemah  \cup \{\eqh\} \cup \funSchema$ where $\relSchemah$ has,  for every $R \in \schema$,  a relation symbol $\Rh$ intended to contain the complement of $R$ and $\eqh$ contains the complement of the relation $=$. This can be achieved as in the proof of \mbox{Lemma \ref{lemma:negationfree}}.
      
      From $\prog'$ we construct a \DynUCQ-program $\prog''$ over schema $\schema'' \df \relSchema \cup \relSchemah \cup \schema_F$ where $\schema_F$ contains a $(k+1)$-ary relation symbol $R_f$ for every $k$-ary function symbol $f \in \funSchema$. The intention is that $R_f$ simulates $f$ in the sense that $(\vec a, b) \in R_f^{\state'}$ if and only if $f^{\state''}(\vec a) = b$ in states $\state'$ and $\state''$ reached in $\prog'$ and $\prog''$ by the same modification sequence. The initialization of $R_f$ can be obtained easily from the initialization of $f$.
      
     We say that two states $\state'$ and $\state''$ over $\schema'$ and $\schema''$ \emph{correspond}, if (1) the condition $(\vec a, b) \in R_f^{\state'}$ if and only if $f^{\state''}(\vec a) = b$ is satisfied, and (2) $R^{\state'} = R^{\state''}$ for all $R \in \relSchema \cup \relSchemah$. 
      
      We explain next how to update relations from $\schema_F$. To this end, we will define $\CQ$-formulas $\varphi_t(\vec x, z)$ and $\varphi_\phi(\vec x)$ over $\schema''$, for every update term $t(\vec x)$ and every update formula $\phi(\vec x)$ over $\schema'$,  such that the following conditions are satisfied for all corresponding states $\state'$, $\state''$, all tuples $\vec a$ and all elements $b$:
      \begin{itemize}
        \item $\state'', (\vec a, b) \models \varphi_t(\vec x, z)$ if and only if $t^{\state'}(\vec a) = b$, and
        \item $\state'', \vec a \models  \varphi_\phi(\vec x)$ if and only if $\state', \vec a \models \phi(\vec x)$
      \end{itemize}
      
      Then the update formulas in $\prog''$ after a modification $\delta$ can be defined as follows. For every $R_f \in \schema_F$, define the update formula as $\uf{R_f}{\delta}{\vec u}{\vec x, y} \df \varphi_t(\vec u, \vec x, y)$ where $t$ is the update term for $f \in \funSchema$ in $\prog'$.  For every $R \in \relSchema \cup \relSchemah$ define the update formula as $\ufwa{R}{\delta} \df \varphi_\phi$ where $\phi$ is the update formula of $R$ in $\prog'$. An easy induction shows that $\prog'$ and $\prog''$ yield corresponding states when the same modification sequence is applied. This proves the claim.
      
      It remains to define the $\CQ$-formulas $\varphi_t(\vec x, z)$ and $\varphi_\phi(\vec x)$ for every update term $t(\vec x)$ and every formula $\phi(\vec x)$. Those formulas are defined inductively as follows:
      \begin{enumerate}
        \item If $t(\vec x) = y$ for some variable $y$ occurring in $\vec x$, then $$\varphi_t(\vec x, z) \df y = z.$$
        \item If $t(\vec x) = f(t_1(\vec x_1), \ldots, t_k(\vec x_k))$ with $\vec x_i \subseteq \vec x$, then $$\varphi_t(\vec x, z) \df \exists z_1 \ldots \exists z_k \Big(R_f(z_1, \ldots, z_k, z) \wedge \bigwedge_{i} \varphi_{t_i}(\vec x_i, z_i) \Big).$$
        \item If $t(\vec x) = \ite{\phi(\vec y)}{t_1(\vec x_1)}{t_2(\vec x_2)}$ with $\vec y, \vec x_1, \vec x_2 \subseteq \vec x$, quantifier-free update formula $\phi$ and update terms $t_1$, $t_2$,  then 
          $$\varphi_t(\vec x, z) \df \big(\varphi_{\phi}(\vec y) \wedge \varphi_{t_1}(\vec x_1, z)\big) \vee \big(\varphi_{\neg \phi}(\vec y) \wedge \varphi_{t_2}(\vec x_2, z)\big).$$
          
        \item If $\phi(\vec x)$ contains the maximal\footnote{Here, a term $t_i$ is maximal if it is not contained in another update term.} update terms $t_1(\vec x_1), \ldots, t_k(\vec x_k)$ then let
        $$\varphi_\phi(\vec x) \df \exists z_1 \ldots \exists z_k \Big(\phi' \wedge \bigwedge_{i} \varphi_{t_i}(\vec x_i, z_i)\Big)$$
      where $\phi'$ is obtained from $\phi$ by replacing $t_i$ by $z_i$, transforming the resulting formula into negation normal form and then replacing every literal of the form $\neg R(s_1, \ldots, s_l)$ by $\Rh(s_1, \ldots, s_l)$.
      \end{enumerate}
      
      Observe that the formula $\phi$ in (d) contains only relation symbols from $\relSchema \cup \relSchemah$, and therefore no relation symbols from $\funSchema$ need to be replaced in $\phi'$. The correctness of this construction can be proved inductively.
      
  \end{proofof}

  \section{$\Delta$-semantics} \label{section:deltasemantics}
      So far we considered a semantics where the new version of the auxiliary relations is redefined, after each modification, from scratch by formulas that are evaluated on the structure with the current auxiliary relations. We refer to this as \emph{absolute semantics} in the following.

  However, in the context of view maintenance, one usually expects only few auxiliary tuples to change after a modification. Therefore it is common to express the new version of the auxiliary relations in terms of the current relations and some ``Delta'', that is,  a (small) relation $R^+$ of tuples to be inserted into $R$ and a (small) relation $R^-$ of tuples to be removed from $R$ (with $R^+ \cap R^- = \emptyset$). The updated auxiliary relation $R'$ is then defined by
    $$R' \df (R \cup R^+) \setminus R^-.$$
  We refer to this semantics as \emph{$\Delta$-semantics}. This is the semantics usually considered in view maintenance. As already stated in the introduction, absolute and  $\Delta$-semantics can only be different if the underlying update language is not closed under Boolean operations.

  Next we formalize $\Delta$-semantics via $\Delta$-update programs which provide formulas defining the relations $R^+$ and $R^-$, for every auxiliary relation $R$.

\begin{definition}($\Delta$-Update program)\label{def:updateprog}
  A \emph{$\Delta$-update program} \prog over dynamic schema \mbox{$(\inpSchema, \auxSchema)$} 
  is a set of first-order formulas (called \textit{$\Delta$-update formulas} in the following) that contains,  for every $R \in \auxSchema$ and every
  $\delta \in \{\ins_S, \del_S\}$ with $S \in \inpSchema$, two formulas  $\uf{R^+}{\delta}{\vec u}{\vec x}$ and $\uf{R^-}{\delta}{\vec u}{\vec x}$ over the schema $\schema$ where $\vec u $ and $S$ have the same arity,  $\vec x$ and $R$ have the same arity, and  $\ufwa{R^+}{\delta}\land \ufwa{R^-}{\delta}$ is unsatisfiable.
\end{definition}

The \emph{semantics of $\Delta$-update programs} is as follows. For a modification $\delta=\delta(\vec a)$ and program state $\state=(\domain, \inp,\aux)$ we denote by $P_\delta(\state)$ the state $(\domain, \delta(\inp), \aux')$, where the relations $R'$ of $\aux'$ are defined by
$$R'\df \Big(R \cup \big\{\vec b \mid \state \models \uf{R^+}{\delta}{\vec a}{\vec b}\big\}\Big) \setminus \big\{\vec b \mid \state \models \uf{R^-}{\delta}{\vec a}{\vec b}\big\}.$$

The effect of a modification sequence on a state, dynamic $\Delta$-programs and so on are defined like their counterparts in absolute semantics except that $\Delta$-update programs are used instead of update programs.

\begin{definition}(\dDynC) \label{definition:dync}
  For a class $\calC$ of formulas, let $\dDynC$ be the class of all dynamic queries that  can be maintained by dynamic
  $\Delta$-programs with formulas from $\calC$ and arbitrary initialization mapping. 
\end{definition}

We note that the definitions above do \emph{not} require that $R^+\cap R=\emptyset$ or $R^-\subseteq R$, that is, $R^+$ might contain tuples that are already in $R$, and $R^-$ might contain tuples that are not in $R$. However, in all proofs below, we construct only $\Delta$-update formulas that guarantee these additional properties. As a consequence, for the considered fragments, the expressive power is independent of this difference.

The goal of this section is to prove the remaining results of Figure~\ref{figure:collapse}, that is, the collapse results depicted in the right part of the figure and the correspondences between absolute semantics and $\Delta$-semantics.

\shortOrLong{}{
  The main results of this section are the following characterizations of (extensions of) dynamic conjunctive queries with $\Delta$-semantics.
}

\begin{theorem}\label{theorem:ddynefoequivalences}
  Let $\query$ be a query. Then the following statements are equivalent:
  \begin{enumerate}
    \item $\query$ can be maintained in $\dDynUCQneg$.
    \item $\query$ can be maintained in $\dDynUCQ$.
    \item $\query$ can be maintained in $\dDynCQneg$.
    \item $\query$ can be maintained in $\dDynCQ$.
    \item $\query$ can be maintained in $\dDynEFO$.
    \item $\query$ can be maintained in $\dDynAFO$.
  \end{enumerate}
\end{theorem}
\begin{theorem}\label{theorem:dynefoandddynefoequivalence}
  Let $\query$ be a query. Then the following statements are equivalent:
  \begin{enumerate}
    \item $\query$ can be maintained in $\DynUCQneg$.
    \item $\query$ can be maintained in $\dDynUCQneg$.
  \end{enumerate}
\end{theorem}

The technique used for removing unions from dynamic unions of conjunctive queries under $\Delta$-semantics can be used to obtain a $\dDynFOand$ normal form for $\dDynFO$-programs.

\begin{theorem}\label{theorem:ddynfoequivalences}
  Let $\query$ be a query. Then the following statements are equivalent:
  \begin{enumerate}
    \item $\query$ can be maintained in $\dDynFO$.
    \item $\query$ can be maintained in $\dDynFOand$.
  \end{enumerate}
\end{theorem}

We state some basic facts about dynamic programs with $\Delta$-semantics before proving those theorems. The following lemma establishes the obvious fact that the absolute semantics and $\Delta$-semantics coincide in expressive power for dynamic classes closed under boolean operations. We observe that the proof does not work for (extensions of) conjunctive queries. Later we will see how to extend the result to conjunctive queries.

\begin{lemma} \label{lemma:absdeltaequivalence}
  Let $\class$ be some fragment of first-order logic closed under the boolean operations $\{\lor, \land, \neg\}$. Then for every query $\query$ the following are equivalent:
  \begin{enumerate}
    \item There is a $\DynC$-program that maintains $\query$.
    \item There is a $\dDynC$-program that maintains $\query$.
  \end{enumerate}  
\end{lemma}
\begin{proof}
    From an $\DynC$-update formula $\ufwa{R}{\delta}$, the $\dDynC$-update formulas are defined as follows:
    \begin{align*}
      \uf{R^+}{\delta}{\vec u}{\vec x}& \df \uf{R}{\delta}{\vec u}{\vec x} \land \neg R(\vec x) \\
      \uf{R^-}{\delta}{\vec u}{\vec x}& \df \neg \uf{R}{\delta}{\vec u}{\vec x} \land R(\vec x)
    \end{align*}
    From a $\dDynC$-update formulas $\ufwa{R^+}{\delta}$ and  $\ufwa{R^+}{\delta}$, an $\DynC$-update formula is obtained via
    $$\uf{R}{\delta}{\vec u}{\vec x} \df \big(R(\vec x) \lor \uf{R^+}{\delta}{\vec u}{\vec x}\big) \wedge \neg \uf{R^-}{\delta}{\vec u}{\vec x}.$$
\end{proof}

Removing negations in dynamic programs with $\Delta$-semantics is straightforward using the replacement technique, since the complement $\Rh$ of an auxiliary relation $R$ can be maintained by exchanging the formulas $\ufwa{R^+}{\delta}$ and $\ufwa{R^-}{\delta}$. Observe that in contrast to absolute semantics this works for arbitrary query classes, even if they are not closed under complementation, and in particular for (extensions of) conjunctive queries.

\begin{lemma}\label{lemma:deltanegationfree}
  Let $\class$ be some fragment of first-order logic. If a query $\query$  can be maintained in $\dDynC$ then $\query$  can be maintained in negation-free $\dDynC$. 
\end{lemma}
\begin{proof} 
  The idea is again to maintain the complements for auxiliary relations. Given a dynamic $\Delta$-program $\prog$ over schema $\schema$ we construct a dynamic $\Delta$-program $\prog'$ over schema $\schema \cup \schemah$ where $\schemah$ contains, for every $k$-ary relation symbol $R \in \schema$, a fresh $k$-ary relation symbol $\Rh$ with the intention that $\Rh$ always stores the complement of $R$. 

  The update formulas for $R \in \schema$ are as in $\prog$. For a relation symbol $R \in \schema$ let $\uf{R^+}{\delta}{\vec u}{\vec x}$ and  $\uf{R^-}{\delta}{\vec u}{\vec x}$ be the update formulas of $R$. Then the update formulas for $\Rh$ can be defined as follows:
  \begin{align*}
    \uf{\Rh^+}{\delta}{\vec u}{\vec x}& = \uf{R^-}{\delta}{\vec u}{\vec x}\\
    \uf{\Rh^-}{\delta}{\vec u}{\vec x}& = \uf{R^+}{\delta}{\vec u}{\vec x}
  \end{align*}

  From $\prog'$, a negation-free dynamic $\Delta$-program $\prog''$ can be constructed by replacing, for  all $R \in \schema$, all occurrences of $\neg R(\vec x)$ in update formulas of $\prog'$ by  $\Rh(\vec x)$. We omit the obvious proof of correctness.
\end{proof}

We now turn towards proving the main results of this section. We first prove Theorem \ref{theorem:dynefoandddynefoequivalence}. Afterwards we  use the connection between absolute and $\Delta$-semantics that it establishes as well as the adaption of Lemma \ref{lemma:disjunctionfree} to $\Delta$-semantics to prove the characterization of conjunctive queries with $\Delta$-semantics.

The only-if-direction of Theorem \ref{theorem:dynefoandddynefoequivalence} can be generalized to arbitrary quantifier prefixes. It is open whether the if-direction generalizes as well.

\begin{lemma} \label{lemma:qfotodeltaqfo}
 Let $\quant$ be an arbitrary quantifier prefix. If a query  can be maintained in $\DynQFO$ then it  can be maintained in $\dDynQFO$ as well.
\end{lemma}
\begin{proof}
  Let $\prog = (P, \init, Q)$ be a $\DynQFO$-program with schema $\schema$. By Lemma \ref{lemma:shifting} we can assume, without loss of generality, that the update formulas of $Q$ are atomic. We construct a dynamic $\dDynQFO$-program \mbox{$\prog' = (P', \init', Q')$}.

 The main challenge is to design update formulas of the kind $\ufwa{R^-}{\delta}$ without being able to complement the given update formulas because this would lead to $\cQFO$-formulas (additionally, the disjointness requirement for formulas $\ufwa{R^+}{\delta}$ needs to be ensured). 

 The basic idea is to  use two copies of the auxiliary relations, both alternating between empty and useful states, such that one copy is useful for even steps and the other one for odd steps. More precisely, for every auxiliary relation $R$ used by $\prog$, the program $\prog'$ uses two auxiliary relations $R_{\even}$ and $R_{\odd}$ with the intention that after an even sequence of modifications $R_{\even}$ stores the content of $R$ after the same sequence of modifications while $R_{\odd}$ is empty. After an odd sequence of modifications $R_{\even}$ is empty while $R_\odd$ stores the content of $R$. 

 Then, for an even modification, the relation $R^+_{\even}$ can be simply expressed as in absolute semantics (using ``odd'' relations) and  $R^-_{\even}$ is empty. For an odd modification $R^-_{\even}$ can be simply chosen as $R_{\even}$ and $R^+_{\even}$ is empty. Similarly for $R_{\odd}$.

 In the following we give a precise construction of $\prog'$ over schema $\schema_{\even} \cup \schema_{\odd} \cup \{\Odd, Q'\}$ where $\Odd$ is a boolean relation symbol, and $\schema_{\even}$ and $\schema_{\odd}$ contain, for every $k$-ary relation symbol $R \in \schema$, a $k$-ary relation symbol $R_{\even}$ and $R_{\odd}$, respectively. The relation $\Odd$ is used to store the parity of the  number of modifications performed so far.

  Let $\ufwa{R}{\delta}$ be the update formula of $R \in \schema$ for a modification $\delta$ in the dynamic program $\prog$. Denote by $\ufsubstitutewa{\schema \rightarrow \schema_{\even}}{R}{\delta}$ the formula obtained from $\ufwa{R}{\delta}$ by replacing every atom $S(\vec x)$ with $S \in \schema$ by $S_{\even}(\vec x)$. Analogously  for $\ufsubstitutewa{\schema \rightarrow \schema_{\odd}}{R}{\delta}$. Now, the update formulas for $R_{\odd}$ and $R_{\even}$ are as follows:
    \begin{align*}
    \uf{R^+_{\odd}}{\delta}{\vec u}{\vec x}& \df \neg \Odd \land \ufsubstitute{\schema \rightarrow \schema_{\even}}{R}{\delta}{\vec u}{\vec x} \\
    \uf{R^-_{\odd}}{\delta}{\vec u}{\vec x}& \df \Odd \land R_{\odd}(\vec x) \\
    \uf{R^+_{\even}}{\delta}{\vec u}{\vec x}& \df \Odd \land \ufsubstitute{\schema \rightarrow \schema_{\odd}}{R}{\delta}{\vec u}{\vec x} \\
    \uf{R^-_{\even}}{\delta}{\vec u}{\vec x}& \df \neg \Odd \land R_{\even}(\vec x)
  \end{align*}

  Observe that all those formulas can be easily converted into $\QFO$-formulas. The boolean auxiliary relation $\Odd$ can be updated straightforwardly.

  Now, since the update formulas of $Q$ in $\prog$ are quantifier-free, the relation $Q'$ can be updated with the following quantifier-free update formulas:
    \begin{align*}
      \uf{Q'^+}{\delta}{\vec u}{\vec x}& \df \uf{Q}{\delta}{\vec u}{\vec x} \land  \\ 
      & \neg \Big( \big (\Odd \land Q_\odd(\vec x) \big) \lor \big(\neg \Odd \land Q_\even(\vec x)\big) \Big) \\
      \uf{Q'^-}{\delta}{\vec u}{\vec x}& \df \neg \uf{Q}{\delta}{\vec u}{\vec x} \land \\
      &  \Big( \big (\Odd \land Q_\odd(\vec x) \big) \lor \big(\neg \Odd \land Q_\even(\vec x)\big) \Big)
    \end{align*}
  
  The initialization mapping of $P'$ is straightforward. Every $R_{\even} \in \schema_{\even}$ is initialized with $\init(R)$. All $R_{\odd} \in \schema_{\odd}$ are initialized with the empty relation. The relation $\Odd$ is initialized with $\bot$, and $Q'$ is initialized with $\init(Q)$. 
\end{proof}

\begin{lemma}\label{lemma:deltaefotoefo}
 \begin{enumerate}
    \item If a query can be maintained in $\dDynUCQneg$ then it  can be maintained in $\DynUCQneg$ as well.
    \item If a query can be maintained in $\dDynAFO$ then it  can be maintained in $\DynAFO$ as well.
 \end{enumerate} 
\end{lemma}
We note that the first statement could equally be expressed in terms of $\dDynEFO$ and $\DynEFO$.

\begin{proof}
  We only prove (a), the proof of (b) is analogous. Let $\prog = (P, \init, Q)$ be a dynamic $\dDynUCQneg$-program over \mbox{schema $\schema$}. By Lemma \ref{lemma:deltanegationfree} we can assume, without loss of generality, that the update formulas of $\prog$ are negation-free. For ease of presentation we assume that the input schema contains a single binary relation symbol $E$.

  We construct an equivalent $\DynUCQneg$-program $\prog'$ using the following idea. Consider some update formulas $\uf{R^+}{\delta}{\vec u}{\vec x}$ and $\uf{R^-}{\delta}{\vec u}{\vec x}$ of a relation $R \in \schema$ for a modification $\delta$ in $\prog$. The na\"{\i}ve translation into a $\DynFO$-update formula $\uf{R}{\delta}{\vec u}{\vec x}$ yields the formula 
  $$\uf{R}{\delta}{\vec u}{\vec x} = (R(\vec x) \lor \uf{R^+}{\delta}{\vec u}{\vec x}) \wedge \neg \uf{R^-}{\delta}{\vec u}{\vec x}$$
  which is possibly non-\UCQneg due to $\neg \uf{R^-}{\delta}{\vec u}{\vec x}$. Therefore, $\prog'$ maintains a relation $R^-_\delta$ that contains all tuples $(\vec a, \vec b)$ such that $\vec a$ would be removed from $R$ after applying the modification $\delta(\vec b)$. Those relations are maintained using the squirrel technique.
  The dynamic program $\prog'$ is over schema $\schema \cup \schema_\Delta$ where $\schema_\Delta$ contains a $(k+2)$-ary relation symbol $R^-_\delta \in \schema$ for every $k$-ary relation symbol $R \in \schema$ and every modification $\delta \in \{\ins, \del\}$ of the input relation $E$.

  The update formula for a relation symbol $R \in \schema$ is 
  $$\uf{R}{\delta}{\vec u}{\vec x} \df (R(\vec x) \lor \uf{R^+}{\delta}{\vec u}{\vec x}) \wedge \neg R^-_\delta(\vec u, \vec x).$$  
   This formula can be translated into an existential formula in a straightforward manner. 

  For updating a relation $R^-_{\delta_1}$ after a modification $\delta_0$, the update formula $\ufwa{R^-}{\delta_1}$ for $R^-$ is used. However, since $R^-_{\delta_1}$ shall store tuples that have to be deleted after applying $\delta_1$, the formula $\ufwa{R^-}{\delta_1}$ has to be adapted to use the content of relation symbols $S \in \schema$ after modification $\delta_0$ (instead, as usual, the content from before the modification). For this purpose relation symbols $S \in \schema$ in $\ufwa{R^-}{\delta_1}$ need to be replaced by their update formulas as defined above. 

  The update formula for $R^-_{\delta_1}$ is 
    $$\uf{R^-_{\delta_1}}{\delta_0}{\vec u_0}{\vec u_1, \vec x} \df \ufsubstitute{\schema \rightarrow \phi^\tau}{R^-_{\delta_1}}{\delta_0}{\vec u_0}{\vec u_1, \vec x}$$
  where $\ufsubstitute{\schema \rightarrow \phi^\tau}{R^-_{\delta_1}}{\delta_0}{\vec u_0}{\vec u_1, \vec x}$ is obtained from $\uf{R^-}{\delta_1}{\vec u}{\vec x}$ by replacing every atom $S(\vec z)$ by $\uf{S}{\delta_0}{\vec u_0}{\vec z}$, as constructed above. Since by our initial assumption, $\ufwa{R^-}{\delta_1}$ itself is an existential formula without negation and all update formulas $\ufwa{S}{\delta_0}$ for $S \in \schema$ are existential, the formula $\ufwa{R^-_{\delta_1}}{\delta_0}$ can be easily converted into an existential formula as well. 
\end{proof}

The following example illustrates the construction of Lemma \ref{lemma:deltaefotoefo}.
\begin{example}\label{example:deltaefotoefo}
  Consider the following negation-free $\Delta$-update formulas for a relation symbol $R$:
    \begin{align*}
      \uf{R^+}{\ins}{u}{x}& = \exists y \big(R(y)  \wedge S(u,x)\big) \\
      \uf{R^-}{\ins}{u}{x}& = \exists y \Big(U(x) \vee \big(R(y) \wedge S(y,u)\big)\Big) \\
      \uf{R^+}{\del}{u}{x}& = \exists y U(y) \\
      \uf{R^-}{\del}{u}{x}& = \exists y \exists z \big(S(x,z) \wedge S(y,u)\big)
    \end{align*}
    
  Then the construction from the previous Lemma \ref{lemma:deltaefotoefo} yields the following update formulas for $R$ and $R^-_{\delta_1}$ which can be easily translated into \UCQneg-formulas:
    \begin{align*}
      \uf{R}{\ins}{u}{x} &= (R(x) \vee \uf{R^+}{\ins}{u}{x}) \wedge \neg R^-_{\ins}(u, x) \\
      \uf{R}{\del}{u}{x} &= (R(x) \vee \uf{R^+}{\del}{u}{x}) \wedge \neg R^-_{\del}(u, x) \\
      \uf{R^-_\ins}{\ins}{u_0}{u_1, x} &= \exists y \Big(\uf{U}{\ins}{u_0}{x}
                   \vee \big(\uf{R}{\ins}{u_0}{y} \wedge \uf{S}{\ins}{u_0}{y, u_1}\big)\Big) \\
      \uf{R^-_\ins}{\del}{u_0}{u_1, x} &= \exists y \Big(\uf{U}{\del}{u_0}{x}\Big) 
                   \vee \big(\uf{R}{\del}{u_0}{y} \wedge \uf{S}{\del}{u_0}{y, u_1}\big) \\
      \uf{R^-_\del}{\ins}{u_0}{u_1, x} &= \exists y \exists z \big(\uf{S}{\ins}{u_0}{x,z} \wedge \uf{S}{\ins}{u_0}{y,u_1}\big) \\
      \uf{R^-_\del}{\del}{u_0}{u_1, x} &= \exists y \exists z \big(\uf{S}{\del}{u_0}{x,z} \wedge \uf{S}{\del}{u_0}{y,u_1}\big)
    \end{align*}
\end{example}

Lemmas \ref{lemma:qfotodeltaqfo} and \ref{lemma:deltaefotoefo} together yield Theorem \ref{theorem:dynefoandddynefoequivalence}. We now finally prove Theorem \ref{theorem:ddynefoequivalences}. For this we need the following adaption of Lemma \ref{lemma:disjunctionfree} to $\Delta$-semantics.

\begin{lemma}\label{lemma:ddisjunctionfree}
  \begin{enumerate}
    \item For every $\dDynUCQneg$-program there is an equivalent $\dDynCQneg$-program.    \item For every $\dDynFO$-program there is an equivalent $\dDynFOand$-program.
  \end{enumerate}
\end{lemma}
\begin{proof}
  The proof uses the idea from the corresponding Lemma \ref{lemma:disjunctionfree} for absolute semantics. We prove (a) only. The construction for (b) is exactly the same. 

  Let $\prog$ be a $\dDynUCQneg$-program. As in Lemma \ref{lemma:disjunctionfree} we construct two programs $\prog'$ and $\prog''$ equivalent to $\prog$ for domains of size at least two and domains of size one, respectively. The construction of $\prog'$ is exactly the same as the construction for absolute semantics. For the construction of $\prog''$,  Lemma \ref{lemma:dnullary} (see below) is used.

  A $\dDynCQneg$-program $\prog'''$ is obtained from $\prog'$ and $\prog''$ by using a modification of the construction used for the cases (b) and (c) in Lemma \ref{lemma:disjunctionfree}. 

  In order to delegate the case distinction to the initialization mapping, we use an additional $0$-ary relation symbol $U$ to ensure that interpretations of relations $R'' \in \schema''$ never change for domains of size a least two and, analogously, interpretations of relations $R' \in \schema'$ never change for domains of size one. 

  To achieve this, $U$ is interpreted by true if and only if the domain is of size at least two and the update formulas of $\prog'$ and $\prog''$ are slightly modified as follows.

  Update formulas $\ufwa{R'^+}{\delta}$ and $\ufwa{R'^-}{\delta}$ of a relation symbol $R' \in \schema'$ in program $\prog'$ are replaced in $\prog'''$ by $\ufwa{R'^+}{\delta} \land U$ and $\ufwa{R'^-}{\delta} \land U$. Hence the interpretation of $R'$ changes only for domains of size at least two.

  Similarly,  update formulas $\ufwa{R''^+}{\delta}$ and $\ufwa{R''^-}{\delta}$ of a relation symbol $R'' \in \schema''$ in program $\prog''$ are replaced in $\prog'''$ by $\ufwa{R''^+}{\delta} \land \neg U$ and $\ufwa{R''^-}{\delta} \land \neg U$. Hence the interpretation of $R''$ changes only for domains of size one.

  The initialization of relation symbols from $\schema' \cup \schema''\cup \{Q'''\}$ is as in \mbox{Lemma \ref{lemma:disjunctionfree}}, and $U$ is initialized as true if and only if $|\domain| = 1$.
\end{proof}

\begin{lemma} \label{lemma:dnullary}
  Every query on a $0$-ary database can be maintained by a $\dDynAND$-program.
\end{lemma}
\begin{proof}
  Let $\inpSchema$ be an input schema with $0$-ary relation symbols $A_1, \ldots, A_k$. Further let $\calQ_1, \ldots, \calQ_{m}$ be an enumeration of all $m = 2^{2^k}$ many queries on $\inpSchema$. As in Lemma \ref{lemma:nullary} we show that all of them can be maintained by one $\dDynAND$-program $\prog$ with auxiliary schema $\auxSchema = \{R_1, \ldots, R_{m}\}$ maintaining $\calQ_i$ in $R_i$, for every $i \in\{1,\ldots,m\}$.

  Our goal is to re-use the program constructed in Lemma \ref{lemma:nullary} and the translation 
  \begin{align*}
    \uf{R_i^+}{\delta}{\vec u}{\vec x}& = \uf{R_i}{\delta}{\vec u}{\vec x} \land \neg R_i(\vec x) \\
    \uf{R_i^-}{\delta}{\vec u}{\vec x}& = \neg \uf{R_i}{\delta}{\vec u}{\vec x} \land R_i(\vec x)
  \end{align*}
  Yet $\neg \ufwa{R_i}{\delta}$ does not yield a $\DynAND$-formulas immediately.

  The idea to solve this issue is to use two dynamic programs $\prog^\land$ and $\prog^\lor$ that both maintain all queries $\query_i$ in their auxiliary relations. The program $\prog^\land$ will be the program from Lemma \ref{lemma:nullary} whereas $\prog^\lor$ will be a \DynProp-program whose update formulas are disjunctions of atoms. Then the update formulas of $R_i$ in $\prog^\lor$ will be used for defining $\ufwa{R_i^-}{\delta}$.

  We make this more precise. By Lemma \ref{lemma:nullary} there is a $\DynAnd$-program $\prog^\land$ over schema $\schema^\land = \{R_1^\land, \ldots, R_m^\land\}$  that maintains $\query_i$ in $R_i^\land$ with conjunctive quantifier-free update formulas. Analogously a dynamic program $\prog^\lor$ over schema $\schema^\lor = \{R_1^\lor, \ldots, R_m^\lor\}$ can be constructed that maintains $\query_i$ in $R_i^\lor$ with disjunctive quantifier-free update formulas.

  Then the update formulas for $R_i$ in $\prog$ are constructed as
  \begin{align*}
      \ufwa{R^+_i}{\delta}& = \ufsubstitutewa{\schema^\land \rightarrow \schema}{R_i^\land}{\delta} \land \neg R(\vec x) \\
      \ufwa{R^-_i}{\delta}& = \neg \ufsubstitutewa{\schema^\lor \rightarrow \schema}{R^\lor_i}{\delta} \land R(\vec x)
  \end{align*}

  where $\ufsubstitutewa{\schema^\land \rightarrow \schema}{R_i^\land}{\delta}$ is obtained from $\ufwa{R_i^\land}{\delta}$ by replacing symbols $S^\land \in \schema^\land$ by $S \in \schema$, and $\ufsubstitutewa{\schema^\lor \rightarrow \schema}{R_i^\lor}{\delta}$ is obtained from $\ufwa{R_i^\lor}{\delta}$ by replacing symbols $S^\lor \in \schema^\lor$ by $S \in \schema$.

  Those update formulas can be easily written as conjunctions. Negations can be removed by Lemma \ref{lemma:deltanegationfree}.
\end{proof}

\begin{proofof}{Theorem  \ref{theorem:ddynefoequivalences}}
  The equivalence of (a) and (b) as well as of (c) and (d) follows from Lemma \ref{lemma:deltanegationfree}. Statements (a) and (c) are equivalent by Lemma \ref{lemma:ddisjunctionfree}. Further, (a) and (e) are equivalent by definition. 
  The equivalence of (e) and (f) follows immediately by combining Lemmas \ref{lemma:qfotodeltaqfo} and \ref{lemma:deltaefotoefo} with Theorem~\ref{theorem:dynefoequivalences}.
\end{proofof}

  \section{A dynamic characterization of first-order logic}\label{section:dyncharact}
    In this section  we characterize first-order queries as the class of queries maintainable by non-recursive $\UCQneg$-programs and, equivalently, by  non-recursive $\DynQFO[\exists^1]$-pro\-grams. Here $\QFO[\exists^1]$ is the class of queries expressible by first-order formulas in prenex normal form  with at most one existential quantifier and no universal quantifiers, and ``non-recursive'' is explained next. 

A dynamic program is \textit{non-recursive} if it has an acyclic dependency graph (as a directed graph). For every class $\calC$,  \emph{non-recursive} $\DynC$ refers to the set of queries that can be maintained by non-recursive $\DynC$-programs.

The objective of this section is to prove the following theorem.
\begin{theorem}\label{theorem:non-recursive}
  For every query $\query$ the following statements are equivalent
  \begin{enumerate}[(a)]
    \item \label{theorem:non-recursivea} $\query$ can be expressed in $\FO$.
    \item \label{theorem:non-recursiveb} $\query$  can be maintained in non-recursive $\DynFO$.
    \item \label{theorem:non-recursivec}$\query$  can be maintained in non-recursive $\DynQFO[\exists^1]$.
    \item \label{theorem:non-recursived}$\query$  can be maintained in non-recursive $\DynQFO[\forall^1]$.
  \end{enumerate}
\end{theorem}
With respect to the number of quantifiers in update formulas this result is optimal because the first-order definable alternating reachability query on graphs of bounded diameter cannot be maintained with quantifier-free update formulas \cite{GeladeMS12}. Theorem~\ref{theorem:non-recursive} should be compared with the result of \cite{GeladeMS12} that all $\EFO$ queries can be maintained in \DynQF.

Combining Theorem \ref{theorem:non-recursive} with Theorem \ref{theorem:dynefoequivalences} immediately yields the following corollary.
\begin{corollary}
  Every first-order query can be maintained in $\DynCQneg$.
\end{corollary}

The rest of this section is devoted to the proof of Theorem \ref{theorem:non-recursive}, more precisely to the equivalence of statements (\ref{theorem:non-recursivea})-(\ref{theorem:non-recursivec}). The equivalence with (\ref{theorem:non-recursived}) follows from Theorem \ref{theorem:dynefoequivalences} and the fact that its proof does not introduce recursion when applied to a non-recursive program.\footnote{Alternatively, the proof of (\ref{theorem:non-recursivea})$\Rightarrow$(\ref{theorem:non-recursivec}) can be easily adapted to show (\ref{theorem:non-recursivea})$\Rightarrow$(\ref{theorem:non-recursived})} It is obvious that (\ref{theorem:non-recursivec}) implies (\ref{theorem:non-recursiveb}). For ease of presentation, we prove the remaining directions (\ref{theorem:non-recursivea})$\Rightarrow$(\ref{theorem:non-recursivec}) and (\ref{theorem:non-recursiveb})$\Rightarrow$(\ref{theorem:non-recursivea}) for the input schema $\inpSchema = \{E\}$ where $E$ is a binary relation symbol. The proofs can be 
easily adapted to general (relational) signatures. 

The proof for \mbox{(\ref{theorem:non-recursivea})$\Rightarrow$(\ref{theorem:non-recursivec})} makes use of the following normal form for \FO.
A formula $\varphi$ is in \emph{existential prefix form} if it has a prefix over $((\neg \exists) | \exists))^*$ and no quantifier occurs after this prefix (e.g.~$\exists x \neg \exists y \neg \big( E(x,x) \rightarrow E(x,y) \big)$ is in existential prefix form with prefix $\exists \neg \exists$). A formula in prefix normal form can be easily translated into existential prefix form by duality of universal and existential quantifiers.  The \emph{prefix length} of a formula in existential normal form is the number of existential and $\neg$-symbols in the maximal prefix ending with $\exists$.

The following example outlines the idea of the construction for the proof \mbox{of (\ref{theorem:non-recursivea})$\Rightarrow$(\ref{theorem:non-recursivec})}. 

\begin{example}
  Consider the query $\query$ defined by 
    \begin{align*}
      \varphi & = \exists x \forall y \big( E(x,x) \rightarrow E(x,y) \big) \\
        & \equiv \exists x \neg \exists y \neg \big( E(x,x) \rightarrow E(x,y) \big)
    \end{align*}
  We construct a non-recursive dynamic $\DynQFO[\exists^1]$-pro\-gram $\prog$ that maintains $\query$ under deletions only (for simplicity). The construction of $\prog$ applies the squirrel technique from Subsection \ref{section:collapse}. It uses a separate auxiliary relation $R_\psi$ for each subformula $\psi$ obtained from $\varphi$ by stripping off a ``quantifier prefix'' from the existential prefix form of $\varphi$. The relation $R_\psi$ reflects the possible states after a sequence of changes whose length equals the number of stripped off $\neg$- and $\exists$-symbols.

In order to update the query relation after the deletion of an edge, we maintain an auxiliary ternary relation\footnote{For simplicity we write $R_1$ instead of $R_{\psi_1}$.} $R_1$ that contains the result of the query \linebreak[4] \mbox{$\psi_1 \df \neg \exists y \neg \big( E(x,x) \rightarrow E(x,y) \big)$} for every choice $a_1$ for $x$ and every (possibly deleted) edge $\vec e_1$, that is $(a_1, \vec e_1) \in R_1$ if and only if
    $$(V, E \setminus \{\vec e_1\}, \{x \mapsto a_1\}) \models \neg \exists y \neg \big(E(x,x) \rightarrow E(x,y)\big).$$
  Then we can define $\uf{Q}{\del}{\vec v_1}{} \df \exists x R_1(x, \vec v_1)$
and it only remains to find a way to update  the relation $R_1$. To this end, we maintain a further relation $R_2$ that contains the result of $\psi_2 \df \exists y \neg \big( E(x,x) \rightarrow E(x,y) \big)$ for every choice $a_1$ for $x$ and all (possibly deleted) edges $\vec e_1, \vec e_2$, that is $(a_1, \vec e_1, \vec e_2) \in R_2$ if and only if
    $$(V, E \setminus \{\vec e_1, \vec e_2\}, \{x \mapsto a_1\}) \models \exists y \neg \big( E(x,x) \rightarrow E(x,y)\big).$$
  Then $\uf{R_1}{\del}{\vec v_1}{x, \vec v_2} \df \neg R_2(x, \vec v_1, \vec v_2)$ and it remains to update the relation $R_2$. Therefore we maintain a relation $R_3$ that contains the result of $\psi_3 = \neg \big( E(x,x) \rightarrow E(x,y) \big)$ for every choice $a_1, a_2$ for $x, y$ and all (possibly deleted) edges $\vec e_1, \vec e_2, \vec e_3$. Then $$\uf{R_2}{\del}{\vec v_1}{x, \vec v_2, \vec v_3} \df \exists y R_3(x, y, \vec v_1, \vec v_2, \vec v_3)$$
  and it remains to update relation $R_3$ via
      $$\uf{R_3}{\del}{\vec v_1}{x,y, \vec v_2, \vec v_3, \vec v_4} \df \neg \big(E'(x,x, \vec v_1, \ldots, \vec v_4) \rightarrow E'(x,y, \vec v_1, \ldots, \vec v_4)\big)$$     
  where $E'$ is the edge relation obtained from $E$ by deleting $\vec v_1$, $\vec v_2$, $\vec v_3$ and $\vec v_4$, that is  $E'(x,y, \vec v_1, \ldots, \vec v_4)$ can be replaced by
    $$ E(x,y) \wedge (x,y) \neq \vec v_1 \wedge \ldots \wedge (x,y) \neq \vec v_4.$$
This completes the description of the program $\prog$ for $\varphi$ which is easily seen to be 
  non-recursive.
\end{example}

    The following definition will be useful in both remaining proofs. For every first-order formula $\varphi$ with $k$ free variables and every sequence $\delta = \delta_1 \ldots \delta_{j}$ over $\{\ins, \del\}$ let $\varphi^E_{\delta_1 \ldots \delta_{j}}$ be a $(k+2j)$-ary formula such that for every graph $G = (V, E)$, every $\vec a \in V^k$ and every instantiation $\alpha = \delta_1(\vec e_1) \ldots \delta_2(\vec e_j)$ of $\delta$ with tuples $\vec e_1, \ldots, \vec e_j \in V^2$:
    $$\text{$\alpha(G) \models \varphi$ if and only if $G \models \varphi_{\delta_1 \ldots \delta_{j}}(\vec a, \vec e_1, \ldots, \vec e_j)$}.$$

  It is straightforward to construct $\varphi^E_{\delta_1 \ldots \delta_{j}}$.    It should be noted that $\varphi^E_{\delta_1 \ldots \delta_{j}}$ can be constructed such that its quantifier-prefix is the same as for $\varphi$. In particular, if
  $\varphi$ is quantifier-free then $\varphi^E_{\delta_1 \ldots \delta_{j}}$ can be constructed quantifier-free as well. For example, if $\delta = \ins \; \del$ and $\varphi(\vec x) = \neg E(\vec x)$ then 
        $$\varphi^E_{\ins \; \del}(\vec x, \vec u_1, \vec u_2) = \neg \Big(\big(E(\vec x) \vee \vec x = \vec u_1\big) \wedge \neg (\vec x = \vec u_2)\Big).$$

  \begin{lemma}\label{lemma:foimpliesdynefo}
      If a query is definable in $\FO$, then it  can be maintained in non-recursive $\DynQFO[\exists^1]$.
  \end{lemma} 
  \begin{proof}

    Inductively over the length of the prefix of a formula $\varphi$ in existential prefix form, we prove that, for every finite sequence $\delta_1 \ldots \delta_{j}$, the query defined by $\varphi_{\delta_1 \ldots \delta_{j}}$ is maintainable in non-recursive $\DynQFO[\exists^1]$. The claim follows by setting $j = 0$. We construct dynamic programs where the result of the query defined by $\varphi_{\delta_1 \ldots \delta_j}$ is stored in the relation $R^\varphi_{\delta_1 \ldots \delta_{j}}$. 

    For a formula $\varphi$ with a prefix of length 0 (i.e.~a quantifier-free formula), we define  
      $$\uf{R^\varphi_{\delta_1 \ldots \delta_{j}}}{\delta_{0}}{\vec v_0}{\vec y, \vec v_1, \ldots, \vec v_j} \df \varphi^E_{\delta_{0} \ldots \delta_{j}}(\vec y, \vec v_0, \ldots, \vec v_j)$$
    where $\varphi^E_{\delta_{0} \ldots \delta_{j}}$ is as defined above  (in the quantifier-free case).

    For the induction step, let $\varphi$ be a formula of prefix length $i$. By induction hypothesis, every query defined by $\psi_{{\delta_1} \ldots \delta_{j}}$ where $\psi$ has prefix length $i-1$, can be maintained in non-recursive $\DynQFO[\exists^1]$, for every sequence $\delta_1 \ldots \delta_{j}$ of modifications. 

    We distinguish the two cases $\varphi(\vec y) = \exists x \psi(x, \vec y)$ and $\varphi(\vec y) = \neg \gamma(\vec y)$. If $\varphi(\vec y) = \exists x \psi(x, \vec y)$ then the dynamic program for $\varphi$ and $\delta_1 \ldots \delta_{j}$ has auxiliary relations $R^\psi_{\delta_0 \ldots \delta_{j}}$ for $\delta_0 \in \{\ins, \del\}$ containing the result of the query $\psi_{{\delta_0} \ldots \delta_{j}}$. Further, 
      $$\uf{R^\varphi_{\delta_1 \ldots \delta_{j}}}{\delta_{0}}{\vec v_0}{\vec y, \vec v_1, \ldots, \vec v_j} \df \exists x R^\psi_{\delta_0 \ldots \delta_{j}}(x, \vec y, \vec v_0, \ldots, \vec v_{j}).$$
    If $\varphi(\vec y) = \neg \gamma(\vec y)$ then the dynamic program for $\varphi$ and $\delta_1 \ldots \delta_{j}$ has auxiliary relations $R^\gamma_{\delta_0 \ldots \delta_{j}}$ for $\delta_0 \in \{\ins, \del\}$ containing the result of the query $\gamma_{{\delta_0} \ldots \delta_{j}}$. Further,
      $$\uf{R^\varphi_{\delta_1 \ldots \delta_{j}}}{\delta_{0}}{\vec v_0}{\vec y, \vec v_1, \ldots, \vec v_j} \df \neg R^\gamma_{\delta_0 \ldots \delta_{j}}(\vec y, \vec v_0, \ldots, \vec v_{j}).$$
    This yields a non-recursive $\QFO[\exists^1]$-program, for every $\varphi_{\delta_1 \ldots \delta_{j}}$. 
  \end{proof}

  We now turn towards proving the implication (\ref{theorem:non-recursiveb})$\Rightarrow$(\ref{theorem:non-recursivea}) in Theorem~\ref{theorem:non-recursive}. The following notion will be useful. A \emph{topological sorting} of a graph $(V, E)$ is a sequence $v_1, \ldots, v_n$ such that every vertex from $V$ occurs exactly once and   $i > j$ for all edges $(v_i, v_j) \in E$. Every acyclic graph has a topological sorting. In particular, if $R_1, \ldots, R_m$ is a topological sorting of the dependency graph of a non-recursive dynamic program $\prog = (P, \init, Q)$ then update formulas for $R_1$ do only contain relation symbols from $\inpSchema$. Further we can assume, without loss of generality, that $R_m = Q$. We say that $R_i$ is on the $i$th level of the dependency graph.

  \begin{lemma} \label{lemma:nonrecdynfotofo}
    If a query  can be maintained in non-recursive $\DynFO$, then it can be expressed in $\FO$.
  \end{lemma}
  \begin{proof}
      Consider a non-recursive dynamic $\DynFO$-program $\prog = (P, \init, Q)$ over input schema $\{E\}$. 
      
      We start with some intuition. Let $R$ be an auxiliary relation of $\prog$ which is (for simplicity) on the first layer of the topological sorting of the dependency graph of $\prog$. That is, the update formulas $\ufwa{R}{\ins E}$ and $\ufwa{R}{\del E}$ of $R$ depend on the input relations only. There is no a priori upper bound on the complexity of the initialization process for $R$. However, after one modification step the relation is redefined via one of the first-order update formulas $\ufwa{R}{\ins E}$ or $\ufwa{R}{\del E}$ which only use atoms over the input relations. Similarly, the auxiliary relations on higher levels of the dependency graph depend in a first-order fashion from the input structure after a constant number of modification steps. This is exploited in the proof.

     More technically, the proof idea is as follows. For every modification pattern $\delta=\delta_1 \ldots \delta_j$ and every auxiliary relation $R$, a first-order formula $\varphi^{R}_\delta$ is constructed that ``precomputes'' the state of $R$ for every possible modification sequence with the pattern $\delta$. Thanks to non-recursiveness, once $\delta$ is longer than the number of auxiliary relations, the formula $\varphi^{R}_\delta$ can only use relations from the input schema. That is, it is just a first-order formula over $\inpSchema$. We get the desired first-order formula for $Q$ by choosing in $\varphi^{R}_\delta$ a sufficiently long modification sequence $\delta$ (by repeatedly inserting and deleting the same tuple).

      We make this more precise now. Let $\query$ be a query which  can be maintained by a non-recursive $\DynFO$-program $\prog = (P, \init, Q)$ over schema $\schema = \inpSchema \cup \auxSchema$. We assume for simplicity that $\inpSchema=\{E\}$, for a binary symbol $E$. We let $R_0 \df E$ and assume that the auxiliary relations $R_1, \ldots, R_m$ are enumerated with respect to a topological sorting of the dependency graph of $\prog$ with $R_m = Q$.

    We define inductively, by $i$, for every sequence $\delta_1\ldots \delta_j$ with
  $j \geq i$, first-order formulas $\varphi^{R_i}_{\delta_1 \ldots \delta_j}(\vec y, \vec x_1, \ldots, \vec x_j)$ over schema $\inpSchema = \{E\}$ such that $\varphi^{R_i}_{\delta_1 \ldots \delta_j}$ defines $R_i$ after modifications $\delta_1(\vec x_1) \ldots \delta_j(\vec x_j)$. More precisely $\varphi^{R_i}_{\delta_1 \ldots \delta_j}$ will be defined such that for every state $\calS = (V, E^\calS, \aux^\calS)$ of $\prog$ and every sequence $\delta = \delta_1(\vec a_1)\ldots \delta_j(\vec a_j)$ of modifications the following holds:
      \begin{equation} \label{eq:formula}
        \updateRelation{\prog}{\delta}{\calS}{R_i} = \{\vec b \mid (V, E) \models \varphi^{R_{i}}_{\delta_1 \ldots \delta_{j}}(\vec b, \vec a_1, \ldots, \vec a_j)\}     
      \end{equation}

  Here $\updateRelation{\prog}{\delta}{\calS}{R_i}$ denotes the relation stored in $R_i$ in state $\updateState{\prog}{\delta}{\calS}$. For $R_0 = E$ the formula $\varphi^{E}_{\delta_1 \ldots \delta_j}$ is as defined before the previous lemma. For $R_i$ with $i \geq 1$ the formula $\varphi^{R_i}_{\delta_1 \ldots \delta_j}(\vec y, \vec x_1, \ldots, \vec x_j)$ is obtained from the update formula $\uf{R_i}{\delta_j}{\vec x_j}{\vec y}$ of $R_i$  by substituting all occurrences of $R_{i'}(\vec z)$ by $\varphi^{R_{i'}}_{\delta_{1} \ldots \delta_{j-1}}(\vec x_{1}, \ldots, \vec x_{j-1}, \vec z)$ for all $i' < i$. Using induction over $i$, one can prove that the formulas $\varphi^{R_i}_{\delta_1 \ldots \delta_j}$ satisfy Equation \ref{eq:formula}.
  As $\prog$ is non-recursive, each formula $\varphi^{R_i}_{\delta_1 \ldots \delta_j}$ with $j\ge i$ is over schema $\{E\}$.

    The first-order formula $\varphi$ for $\query$ over schema $\inpSchema = \{E\}$ can be constructed as follows. The formula ``guesses'' a tuple $\vec a \in E$, deletes and inserts it $m$ times and applies  $\varphi^{R_m}_{(\del \; \ins)^m}$ to the result (which is identical to the current graph), or (for the case that $E$ is empty) it guesses  a tuple $\vec a \not\in E$, inserts and deletes it $m$ times and applies  $\varphi^{R_m}_{( \ins\;\del)^m}$ to the result.

  More precisely, $\varphi$ for  $\query$ is defined by
      \begin{multline*}
        \varphi(\vec y) \df \exists \vec x \big( (E(\vec x) \wedge \varphi^{R_m}_{(\del \; \ins)^m}(\vec y, \underbrace{\vec x,\vec x, \ldots, \vec x}_{2m-\text{times}})) \\
        \vee  (\neg E(\vec x) \wedge \varphi^{R_m}_{(\ins \; \del)^m}(\underbrace{\vec y, \vec x,\vec x, \ldots, \vec x}_{2m-\text{times}}) )\big).
      \end{multline*}\end{proof}

  \section{Discussion and Future Work}\label{section:conclusion}
     
In this work, we studied dynamic conjunctive queries. We have shown that, contrary to the static setting, many fragments collapse in the dynamic world. Further we proved that $\DynCQ$ captures $\DynQF$ which implies that $\DynCQ$ is strictly larger than $\DynProp$. Moreover a close connection between absolute semantics and $\Delta$-semantics for conjunctive que\-ries has been established. These results were summarized in \mbox{Figure \ref{figure:collapse}}. Finally, it has been shown that dynamic conjunctive queries with negations capture (static) first-order logic.

All results are for arbitrary initialization mappings. However, they also hold in the setting with first-order definable initialization mappings. They do not carry over when the initialization mapping and updates have to be definable in the same class.

Although we have a good picture of the relationship of the various fragments now, it remains open whether the remaining classes $\DynCQ$, $\DynCQneg$ and $\DynFO$ can be separated or collapsed.

In addition to untangling the remaining variations of conjunctive queries, the dynamic quantifier hierarchy and quantifier alternation hierarchy, respectively, deserve a closer look. Lemma \ref{lemma:quantifierswitch} shows that in the dynamic setting the $\Sigma_i$- and $\Pi_i$-fragment of first-order logic coincide. Whether there is a strict $\Sigma_i$-hierarchy remains open. Furthermore, the equivalence of $\EFO$ with absolute and $\Delta$-semantics does not immediately translate to fragments of $\FO$ with alternating quantifiers (although one of the direction does, see \mbox{Lemma \ref{lemma:qfotodeltaqfo}}).

Capturing first-order logic by dynamic conjunctive queries with negations does not immediately yield performance gains (since a first-order query with $k$ quantifiers is translated to a dynamic $\DynCQneg$-program of arity at least $k$). In future work we plan to study whether the work that has been started here can be used to improve the performance of query maintenance.

 \bibliographystyle{plain}
 \bibliography{bibliography}

\end{document}